%% file: arxiv.tex
\documentclass[]{article} 

\usepackage{balance} 

\usepackage{nicefrac}
\usepackage{amsfonts}       
\usepackage{amsmath}
\usepackage{amsthm}
\usepackage{tikz}
\usepackage{xcolor}
\usepackage{comment}
\usepackage{bbm}
\usepackage{natbib}
\usepackage{thmtools}  
\usepackage{thm-restate}
\usepackage{subfig}
\usepackage{todonotes}

\newcommand{\OutN}[1]{\mathtt{NB}_{out}(#1)}
\usepackage{pgfplots}
\DeclareRobustCommand{\citet}[1]{\citeauthor{#1}~\citeyear{#1}}

\title{Measuring a Priori Voting Power - Taking Delegations Seriously}

\author{Rachael Colley,$^1$ Théo Delemazure,$^{2}$ Hugo Gilbert$^2$}

\date{rachael.colley@irit.fr, theo.delemazure@dauphine.eu, hugo.gilbert@lamsade.dauphine.fr\\
$^1$ IRIT, Université Toulouse Capitole, Toulouse, France\\
$^2$Université Paris Dauphine, PSL Research University, CNRS, Lamsade, 75016 Paris, France}

\newcommand{\BibTeX}{\rm B\kern-.05em{\sc i\kern-.025em b}\kern-.08em\TeX}

\newtheorem{definition}{Definition}
\newtheorem{theorem}{Theorem}

\newtheorem{proposition}{Proposition}
\newtheorem{example}{Example}

\newcommand{\M}[1]{\mathcal{M}^{#1}}

\begin{document}




\maketitle

\begin{abstract}
We introduce new power indices to measure the a priori voting power of voters in liquid democracy elections where an underlying network restricts delegations. We argue that our power indices are natural extensions of the standard Penrose-Banzhaf index in simple voting games. 
We show that computing the criticality of a voter is \#P-hard even when voting weights are polynomially-bounded in the size of the instance. 
However, for specific settings, such as when the underlying network is a bipartite or complete graph, recursive formulas can compute these indices for weighted voting games in pseudo-polynomial time. 
We highlight their theoretical properties and provide numerical results to illustrate how restricting the possible delegations can alter voters' voting power.
\end{abstract}



\section{Introduction}
Voting games have been used extensively  to study the a priori voting power of voters participating in an election~\citep{felsenthal1998measurement}. 
A priori voting power means the power granted solely by the rules governing the election process. 
Notably, these measures do not consider the nature of the bill nor the affinities between voters. 
The class of I-power measures (where I represents influence)  ``calculate'' how likely a voter will influence the outcome. 
Several I-power measures have been defined, the best known being the Penrose-Banzhaf measure in simple voting games~\citep{banzhaf1964weighted,penrose1946elementary}. 
In simple voting games, an assembly of voters make a collective decision on a proposal which voters  either support or oppose. 
The Penrose-Banzhaf measure is as follows: voters are assumed to vote independently from one another; a voter is as likely to vote in favour or against the proposal.
It then measures the probability of a voter altering the election's outcome given this probabilistic model. 

Simple voting games have been extended in several directions to take into account more complex and realistic frameworks.  %
For example, taking into account abstention~\citep{freixas2012probabilistic}, several levels of approval~\citep{freixas2003weighted}, or coalition structures~\citep{owen1981modification}.  
Hence, new power indices have been designed to better understand voters' criticality in these frameworks. 
However, elections with delegations have been largely unexplored with respect to a priori voting power. 
Yet, frameworks such as Proxy Voting (PV)~\citep{miller1969program,tullock1992computerizing} or Liquid Democracy (LD)~\citep{behrens2014principles,brill2018interactive} have received increasing interest in the AI community due to their ability to provide more flexibility and engagement in the voting process. 
Thus, studying these frameworks via their distribution of a priori voting power is an interesting research direction. 

\paragraph{Our contribution.} 
We extend simple voting games to model elections where voters can delegate their votes through a social network, modelled as a digraph $G$. 
Our model encapsulates both the LD and PV settings. 
We design a new I-voting power measure to measure voters' criticality in  these settings. 
We argue that our power measure is a natural extension of the Penrose-Banzhaf measure, and we illustrate the intuitions behind it through various examples. 
When $G$ is an arbitrary digraph, we show that the computation of our measure is $\#$P-hard even when voters' weights are polynomially-bounded in the number of voters.
However, we prove that it can be computed for weighted voting games in pseudo-polynomial time in the PV setting, in which the graph $G$ is directed bipartite with all arcs going from one side (possible delegators) to the other (proxies), and in the LD setting when $G$ is complete. 
Last, we complement our theoretical results with numerical results to illustrate how introducing delegations  modifies voters' a priori voting power.

\subsection*{Related Work}

\paragraph{Voting power.} Measuring a voter's voting power in a specific setting quantifies how \emph{critical} they are in deciding the outcome of the election. 
A voter $i$'s voting power can be considered as the difference in probability of $i$ voting for the issue when the outcome is also in favour and $i$ voting for the issue when the outcome is not \citep{gelman2002mathematics}. We give an overview of some standard measures  (we recommend \citep{lucas1974measuring} for an overview of voting power and \citep{felsenthal2005voting} for a historical overview).

The measure introduced by \citeauthor{shapley1954method}~\citeyearpar{shapley1954method} quantifies the voter's expected pay-off, known as P-power, unlike the other measures we will discuss. P-power differs in motivation from I-power as P-power cares about the voter being part of the winning coalition, sharing the coalition's utility among its members, with those not in the coalition receiving utility $0$. In contrast, I-power has a policy-seeking motivation and is, therefore, concerned with the voter's stance on the issue. 

I-power was independently given a mathematical explanation by  \citeauthor{penrose1946elementary}~\citeyearpar{penrose1946elementary}, \citeauthor{banzhaf1964weighted}~\citeyearpar{banzhaf1964weighted}, and \citeauthor{coleman1971control}~\citeyearpar{coleman1971control}. 
It counts, for an agent $i$, in how many of the $2^n$ possible voting profiles that changing $i$'s vote from $0$ to $1$ changes the outcome. The \emph{Banzhaf measure} (or absolute Banzhaf index) is denoted by $\beta'$, whereas the \emph{Banzhaf index} is the relative quantity denoted by $\beta$ (found from normalising $\beta'$). 

\paragraph{Extending the notion of voting power.} Standard voting power measures are defined on binary issues. Yet, as the study of voting models has advanced, so has the study of voting power. One generalisation is to the domain of the available votes, thus, moving away from binary decisions on the issue. 
Influenced by \citeauthor{felsenthal2001models}~\citeyearpar{felsenthal2001models},  probabilistic models of voting power with abstention and a binary outcome are well-studied~\citep{felsenthal1997ternary,freixas2012probabilistic}.  \citeauthor{freixas2016power}~\citeyearpar{freixas2016power} extended the Banzhaf index by introducing two measures of being positively critical, i.e., changing your vote from being for the issue to abstaining and an abstaining vote to be against the issue. Voting games with approvals form a subclass of voting games with varying levels of approvals in both the input and output of the election~\citep{freixas2003weighted}.
Another well-studied extension of the standard notions of voting power measures in WVGs is to allow for randomised weights. The Shapley-Shubik index has been well-studied in these settings \citep{filmus2016shapley,bachrach2016characterization}. \citeauthor{boratyn2020average}~\citeyearpar{boratyn2020average} also studied the Banzhaf index in this setting; this is close to our own when focussing on proxy voting elections. 

\paragraph{Proxy voting (PV) and liquid democracy (LD).} PV allows agents to choose their proxy from a list of representatives who will vote on their behalf. In some models, a delegator may only choose a proxy from the list of representatives~\citep{abramowitz2018flexible,alger2006voting,CohensiusMMMO17}. In other models, delegators can also vote directly yet still cannot receive votes \citep{green2015direct,miller1969program,tullock1992computerizing}. 

Models of LD allow voters to either vote on the issue or delegate their vote to another voter, which can be transitively delegated further. Some examples of recent advancements in the study of LD are: extending the model to account for different situations, whether it be ranked delegations \citep{brill2021liquid,colleyunrav,kotsialou2018incentivising} or allowing for multiple interconnected issues \citep{brill2018pairwise,jain2021preserving}; assessing how successful LD is in finding a ground truth \citep{halpern2021defense,kahng2018liquid}; 
or studying (non-cooperative) game-theoretic aspects \citep{bloembergen2019rational,EscoffierGP20,MarkakisP21,NoelSV21}.  

Our closest work is that of \citeauthor{zhang2021power}~\citeyearpar{zhang2021power},  who study a version of the Banzhaf measure in LD. 
Their measure, for a given delegation graph, determines how critical an agent is in changing the outcome. Our work differs as we focus on \emph{a priori} voting power, where no prior knowledge is known about the election, such as a specific delegation graph.\footnote{In Appendix~\ref{subsec:zhangGrossi}, we give a probabilistic model where the Banzhaf measure from \citeauthor{zhang2021power}~\citeyearpar{zhang2021power} can be interpreted as an I-power measure.}

\section{Model}

Let $V$ be a set of $n$ voters taking part in an election to decide if some binary proposal should be accepted or not. 
Each voter has different possible actions: they may vote directly, either for $(1)$ or against $(-1)$ the proposal or delegate their vote to another voter.  
A voter who decides to vote (resp. delegate) will be termed a delegatee (resp. delegator). 
An underlying social network $G=(V,E)$  
restricts the possible delegations between the agents, hence voter $i\in V$ can only delegate to a voter in their out-neighbourhood $\OutN{i} = \{j\in V\mid  (i,j)\in E\}$. We will consider in more detail two cases: when $G$ is complete and when $G$ is bipartite, where the former corresponds to the LD setting when voters can choose any other voter as a delegate and the latter corresponds to PV. 

\begin{definition} 
Given a digraph $G = (V,E)$, a \emph{G-delegation partition} $D$ is a map defined on $V$ such that $D(i)\in \OutN{i}\cup\{-1,1\}$ for all $i\in V$.  We let $\mathcal{D}$ be the set of all such partitions and $D^-$, $D^+$, and $D^v$ be the inverse images of $\{-1\}$, $\{1\}$ and $\{v\}$ for each $v\in V$ under $D$.
\end{definition}

Whereas a direct-vote partition divides the voters such that each partition cell corresponds to a possible voting option. 
We allow for abstentions to model  situations in which a delegator does not have a delegatee voting on their behalf (e.g., due to delegation cycles).

\begin{definition}
A \emph{direct-vote partition} of a set $V$ is a map $T$ from $V$ to the votes $\{-1,0,1\}$. We let $T^-$, $T^0$, and $T^+$ denote the inverse images of $\{-1\}$, $\{0\}$ and $\{1\}$ under $T$.
\end{definition}

A $G$-delegation partition $D$ naturally induces a direct-vote partition $T_D$ by resolving the delegations. 
First, we let voters in $D^-$, and $D^+$ also be in  $T^-$, and $T^+$, respectively. 
From this point, for some $\circ\in\{-,+\}$, if $v'\in D^v$ and $v\in T^\circ$, then $v'\in T^\circ$. This continues until no more voters can be added to $T^+$ or $T^-$. The remaining unassigned agents abstain and thus are in $T^0$. This procedure assigns agents their delegate's vote unless it leads to a cycle, in this case, their vote is recorded as an abstention. 

Next we define a partial ordering $\leq$ among direct-vote partitions: if $T_1$ and $T_2$ are two direct-vote partitions of $V$, we let:
$T_1 \leq T_2 \Leftrightarrow T_1(x) \le T_2(x), \forall x\in V$.

\begin{definition}
A ternary (resp. binary) voting rule is a map $W$ from the set $\{-1,0,1\}^n$ (resp. $\{-1,1\}^n$) of all direct-vote partitions (resp. all direct-vote partitions without abstention) of $V$ to $\{-1,1\}$ satisfying the following conditions:
\begin{enumerate}
    \item $W(\mathbbm{1}) =  1$ and $W(-\mathbbm{1}) = -1$ where $\mathbbm{1} = (\underbrace{1,\ldots,1}_{\times n})$;
    \item Monotonicity: $T_1 \le T_2 \Rightarrow W(T_1) \le W(T_2)$.\footnote{All ternary voting rules do not satisfy monotonicity, e.g., a weighted voting rule with an additional quorum condition. However, we enforce this condition such that we may only look at the election result when the voter favours the proposal on the one hand and against the proposal on the other to define criticality.}
\end{enumerate}
\end{definition}

Note that ternary (and binary) voting rules only use the direct-vote partition to find an outcome, i.e., only the information of which agents voted directly or indirectly in favour or against the proposal or abstained. 
Thus, these rules do not need the delegations to find an outcome. 

A ternary (resp. binary) voting rule is \emph{symmetric} if $W(T) = - W(-T)$, where $-T$ is the direct-vote partition defined by $-T(x) = -(T(x))$, $\forall x\in V$. 
Moreover, for ease of notation, we may also use $W(T^+, T^-)$ to denote $W(T)$, noting that $T^0$ can be obtained from $T^+$ and $T^-$.

\paragraph{Weighted voting games.} 
Weighted Voting Games (WVGs) express ternary voting rules compactly, with  a quota $q\in (0.5,1]$ and a map $w : V \rightarrow \mathbbm{N}_{>0}$ assigning each voter a positive weight. 
Given a set $S\subseteq V$, we let $w(S) = \sum_{i \in S} w(i)$. In a WVG with weight function $w$, we let $W(T) = 1$ iff $w(T^+) > q\times w(T^+\cup T^-) $, i.e., the proposal is accepted if the sum of the voters' weights for the proposal is greater than a proportion $q$ of the total weight of non-abstaining voters; otherwise, the proposal is rejected. 

\begin{example}\label{ex:introunderlying}
  Consider agents $V=\{a, b, \cdots, m\}$ connected by the underlying network $G$ as depicted in Figure~\ref{fig:introexampleunderlying}. The solid lines give a valid $G$-delegation partition $D$ with $D^+=\{c,i\}$,  $D^-=\{a,b,h\}$, $D^a=\{d\}$, $D^b=\{e\}$, $D^c=\{f,g\}$, $D^d=\{j,k\}$, $D^\ell=\{m\}$, and $D^m=\{\ell\}$. 
  This $G$-delegation partition induces the following direct-vote partition: $T^+=\{c,f,g,i\}$, $T^-=\{a,b,d,e,h,j,k\}$, and $T^0=\{\ell, m \}$.
 Consider $W$ induced by the WVG where $q= 0.5$ and $w(a) = 3$, $w(b)=w(d)= 2$ and the remaining voters $x\in V\backslash \{a,b,d\}$ have weight $w(x)=1$. The proposal is rejected under this $G$-delegation partition as $w(T^+)= 4 \le 7.5= q \cdot w(T^+\cup T^-) = \nicefrac{15}{2}$.
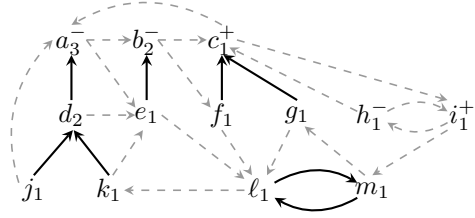
\begin{figure}[t]
    \centering
\begin{tikzpicture}[node distance=0.25cm,semithick]
\node at  (0,0) {$j_1$};
\node at  (1,0) {$k_1$};
\node at  (3,0) {$\ell_1$};
\node at  (4.5,0) {$m_1$};
\node at  (0.5,1) {$d_2$};
\node at  (1.5,1) {$e_1$};
\node at  (2.5,1) {$f_1$};
\node at  (3.5,1) {$g_1$};
\node at  (4.5,1) {$h_1^-$};
\node at  (5.7,1) {$i_1^+$};
\node at  (0.5,2) {$a_3^-$};
\node at  (1.5,2) {$b^-_2$};
\node at  (2.5,2) {$c^+_1$};


\draw[-stealth, thick] (3.2, 0.1) edge[bend left]  (4.3,0.1)
(4.3, -0.1) edge[bend left]  (3.2,-0.1)
(0.5, 1.2) edge  (0.5,1.8)
(1.5, 1.2) edge  (1.5,1.8)
(2.5, 1.2) edge  (2.5,1.8)
(3.5, 1.2) edge  (2.52,1.8)
(1, 0.2) edge  (0.53,0.8)
(0, 0.2) to  (0.52,0.8);

\draw[-stealth, dashed, color=black!40] (-0.2, 0.1) edge[bend left]  (0.25,2) 
(2.5, 0.8) edge  (2.9,0.2) 
(3.4, 0.8) edge  (3.1,0.2) 
(4.7, 1.05) edge[bend left]  (5.5,1.05)  
(5.5,0.9) edge[bend left]  (4.7, 0.9) 
(5.5, 0.8) edge  (4.5,0.2) 
(2.8, 0) edge  (1.2,0) 
(0.7, 2) edge  (1.3,2) 
(0.7, 2) edge  (1.35,1.1) 
(1.65, 2) edge  (2.3,2) 
(1.65, 2) edge  (2.35,1.1) 
(2.7, 2) edge  (5.5,1.2) 
(2.5, 2.2) edge[bend right]  (0.5,2.2) 
(0.7, 1) edge  (1.35,1) 
(1.7, 1) edge  (2.8,0.15)  
(4.3, 1.05) edge  (2.6,1.9)  
(4.4,0.2) edge  (3.6,0.85) 
(1.05, 0.2) to  (1.45,0.8); 
\end{tikzpicture} 
    \caption{\small The underlying network $G$ used in Example~\ref{ex:introunderlying}. While all edges give us $E$, the solid edges give us a valid $G$-delegation partition where the superscripts of $\boldsymbol{+}$ or $\boldsymbol{-}$ represent delegatees direct votes. Each node's subscript refers to its voting weight. }
    \label{fig:introexampleunderlying}
\end{figure}
\end{example}

We conclude this subsection with some notation.  
Given a set $X$, let $\mathcal{P}_k(X)$ denote the set of $k$ ordered partitions of $X$. 
By ordered partitions, we mean that $(\{1\},\{2,3\})$ 
should be considered different to $(\{2,3\},\{1\})$. 
Next, given a voting rule $W$, a voter $i\in V$, and $(X,Y,Z)$ three non-intersecting subsets of $V\setminus\{i\}$, we define: 
$$
\delta_{i,-\rightarrow +}^{W}(X,Y,Z ) \!=\! \frac{W( S\!\cup\! U\!\cup\!\{i\} ,T) - W( X,Y \!\cup\! Z \!\cup\!\{i\})}{2}.
$$ 
We say a voter $i\in V$ is critical when they can affect the outcome of the vote. Thus for three non-intersecting subsets of $V\setminus\{i\}$, namely $X,Y,Z$, where $X$ (resp. $Y$) denotes the set of voters opposing (resp. supporting) the proposal through their vote of delegation, and $Z$ is the set of voters delegating directly or not to $i$, then $i$ is critical if and only if $\delta_{i,-\rightarrow +}^{W}(X,Y,Z)> 0$. 
We say a voter $i\in V$ is \emph{positively} (resp, \emph{negatively}) critical if by changing a positive (resp. negative) vote to a negative (resp. positive) one, the outcome will also change from being for to against (resp. against to for) the issue. 
In Example~\ref{ex:introunderlying}, we see that $a$ is critical in this $G$-delegation partition, as $V\backslash\{a\}$ is partitioned as such $X=\{c,f,g,i\}$, $Y=\{b,e,h\}$ and $Z=\{d,j,k\}$ and thus $\delta_{a,-\rightarrow +}^{W}(X,Y,Z)=\frac{1-(-1)}{2}=1$. 

\subsection{Modelling a priori voting power}
We aim to measure \emph{a priori} voting power in this setting.  An agent's voting power is their probability of being able to affect the election's outcome. Similarly to the intuitions behind the Penrose-Banzhaf measure, we invoke the principle of insufficient reason. There are two ways of seeing this principle. 

\paragraph{The global uniformity assumption.}
If there is no information about the proposal or voters, we assume all $G$-delegation partitions are equally likely  with probability $\Pi_{i\in V} \frac{1}{|\OutN{i}|+2}$. In Example~\ref{ex:introunderlying}, as $|\OutN{i}|=2$ for every $i\in V$, this means that every $G$-delegation partition occurs with probability $(\frac{1}{4})^{13}$.

\paragraph{The individual uniformity assumption.}
The global uniformity assumption is similar to a model in which each voter delegates with probability $p_d^{i} = \nicefrac{|\OutN{i}|}{|\OutN{i}|+2}$ and votes with probability $p_v^{i} = 1 - p_d^{i} =  \nicefrac{2}{|\OutN{i}|+2}$. 
Delegation (resp. voting) options are chosen uniformly at random and voters make their choices independently from one another. 
This is consistent with the idea that we have no information about voters' personalities and interests, or the nature of the proposal. Hence, voters should be equally likely to support (probability $p_y$) or oppose (probability $p_n$) the proposal, i.e., $p_y = p_n = 1/2$.  
Moreover, in ignorance of any concurrence or opposition of interests between voters, we should assume that the likelihood of a voter choosing between each of their possible delegates is equally likely, i.e., the probability that a delegator $i$ delegates to a voter $j\in \OutN{i}$ is $1/\OutN{i}$.    
The \emph{individual uniformity assumption} is an extension of the global uniformity assumption in which  $p_d^i$ can be any value in $[0,1]$ dependent on $|\OutN{i}|$, such that $p_d^i = 0$ when $\OutN{i} = \emptyset$.

For generality, we consider this latter model unless specified otherwise. 
We now define the LD Penrose-Banzhaf measure of a voter $i$ for a given underlying graph $G$ when considering that the probability of each 
 $G$-delegation partition is determined by the  individual uniformity assumption. 

\begin{definition}[\emph{LD} Penrose-Banzhaf measure]
Given a digraph $G=(V,E)$ and a ternary voting rule $W$, the \emph{LD Penrose-Banzhaf measure} of voter $i\in V$ is defined as:
\begin{align*}
\M{ld}_i(W,G) =  \sum_{D\in \mathcal{D}} \mathbbm{P}(D) \frac{W( T_{D^+_i}) - W( T_{D^-_i})}{2},
\end{align*}
where $\mathbbm{P}(D)$ is the probability of the $G$-delegation partition $D$ occurring, and $D^+_i$ (resp. $D^-_i$) is the $G$-delegation partition identical to $D$ with the only possible difference being that $i$ supports (resp. opposes) the proposal.
\end{definition}

$\M{ld}_i$ quantifies the probability to sample a delegation partition where $i$ is able to alter the election's outcome (formally stated in the following Theorem).

\begin{restatable}{theorem}{theoremLDMeasure}\label{thrm : LD Measure}
Given a digraph $G=(V,E)$, a ternary voting rule $W$, and a voter $i\in V$, we have that:
\begin{align*}
\mathbbm{P}(\text{i is critical}) &= \M{ld}_i(W,G).
\end{align*}
Moreover, if $\OutN{i} = \emptyset$ or $W$ is symmetrical, we have that:
\begin{align*}
\mathbbm{P}(\text{i is positively critical}) &= \M{ld}_i(W,G)/2  \\
 &=\mathbbm{P}(\text{i is negatively critical}).    
\end{align*}
\end{restatable}
This proof relies on the fact that we are summing over the probability of each $D$  with respect to $W(T_{D^+_i}) - W( T_{D^-_i})$, which measures when the voter $i$ is critical. Recall that being positively critical means that changing a vote for to against the issue will also change the outcome in the same way (negatively critical is defined similarly). Furthermore, this happens equally when $\OutN{i}=\emptyset$ (the only option is to vote either for or against the issue) or when $W$ is symmetric.

 For the second part of Theorem~\ref{thrm : LD Measure}, the condition is necessary as if $W$ reflects unanimity, i.e., $W(T) = 1 $ iff $T = \mathbbm{1}$, then voters will be more likely to be positively critical than negatively critical.~\footnote{If the voting rule requires total agreement to accept the proposal, then voter $i$ will be critical iff all other voters agree on the proposal. Thus, the probability that $i$  is critical while voting directly or indirectly in favour of the proposal is higher than $i$ being critical while voting directly or indirectly against the proposal.} 
Additionally, observe that the \emph{LD} Penrose-Banzhaf measure of voting power extends the standard Penrose-Banzhaf measure (formalized in Proposition~\ref{prop : extend BP}) and that its values are not normalized (i.e., summing over the agents does not yield 1). The corresponding voting power index can be defined by normalizing over voters. 

\begin{proposition}\label{prop : extend BP}
If $p_d^i = 0$ for all $i\in V$, e.g., if $E = \emptyset$, then the \emph{LD} Penrose-Banzhaf measure of voting power is equivalent to the standard Penrose-Banzhaf voting power measure. 
\end{proposition}

\section{Hardness of computation}
Computing the standard Penrose-Banzhaf measure 
 in WVGs is $\#P$-complete \citep{prasad1990np}. However, it can be computed by a pseudo-polynomial algorithm that runs in polynomial time w.r.t. the number of voters and the maximum weight of a voter~\citep{matsui2000survey}. 
We show that the problem of computing the  \emph{LD} Penrose-Banzhaf measure is $\#P$-hard even when voter's weights are bounded linearly by the number of voters. Hence, a similar pseudo-polynomial algorithm is unlikely to exist. 
The proof uses an enumeration trick inspired by that of \citeauthor{chen2010scalable}~[\citeyear{chen2010scalable}, Theorem~1]. 
Informally speaking, this trick shows that one can solve the $\#P$-hard problem of counting the number of simple paths between two vertices in a digraph by using a polynomial number of calls to a subroutine solving our power measure computation problem, and then inverting a specific Vandermonde matrix. Hence, note that the type of reduction that is used is a Turing reduction.  

\begin{restatable}{theorem}{theoremHardness} 
Given a digraph $G = (V,E)$ and a WVG defined on $V$, computing the \emph{LD}-Penrose-Banzhaf power measure of a voter is $\#P$-hard under Turing reductions even when voter's weights are bounded linearly by the number of voters.
\end{restatable}
\begin{proof}[Proof sketch.] 
We give a reduction from the problem of counting simple paths in a digraph which is known to be \#P-complete \citep{valiant1979complexity}.  
The problem takes as input a digraph $G = (V,E)$ and nodes $s,t \in V$. The problem then returns the number of simple paths from $s$ to $t$ in $G$. Let denote $\mathcal P_\ell$ the set of paths of length $\ell$ between $s$ and $t$ in $G$. 
Given $G = (V,E)$ and $s,t \in V$ two vertices, we create $|V|+1$ different digraphs $G_k = (V_k, E_k)$ with $k \in \{0,\ldots |V|+1\}$ such that $G_k$ is obtained by modifying $G$ to impose some condition on the out-degree of nodes in $V$. 
Thus, in each digraph $G_k = (V_k, E_k)$, we consider a WVG where weights are linearly bounded in $|V|$ and such that voter $s$ is a dictator. 
Hence, $t$ is only critical when in $s$'s delegation path. 
Under the individual uniformity assumption, 
 we obtain the criticality of $t$ in each $G_k$ as a weighted sum of values $|\mathcal{P}_\ell|$ such that the weights of these $|V|+1$ equations form a Vandermonde matrix. Inverting this matrix makes it possible to derive the values $|\mathcal{P}_\ell|$ from the criticality of $t$ in each graph $G_k$, and thus to solve the problem of counting simple paths from $s$ to $t$.
\end{proof}

While computing exactly the \emph{LD} Penrose-Banzhaf measures of voting power is hard, these values can be approximated easily using a standard sampling procedure. 
We sample enough $G$-delegation partitions by simulating the behaviours of the different voters according to the individual uniformity assumption and consider the expected criticality of the voters given these samples. 
Relying on Hoeffding's inequality, one can then prove that these estimates are within some $\epsilon$ of the true voting power measure. 

In the next two sections, we discuss two restricted classes of instances for which more compact formulations of the LD Penrose-Banzhaf measure can be designed such that there exists a pseudo-polynomial algorithm. 

\section{Proxy voting} \label{sec:bipartite}

This section models a PV setting where  $G = (V,E)$ is bipartite with $V = (V_d, V_v)$ and $E = \{(i,j) \mid i \in V_d, j \in V_v \}$. 
The set of delegatees $V_v$ is given in input and is predetermined, e.g., by an election, self-nomination, or sortition. Each delegatee $i\in V_v$ will vote, i.e., $\OutN{i} = \emptyset$ and $p_d=0$, 
whereas each voter $i \in V_d$ can vote or delegate to any delegatee in $V_v$, i.e., $\OutN{i} = V_v$.\footnote{We present an alternative PV model in Appendix~\ref{appendix:PValpha} where voters in $V_d$ must delegate to a voter in $V_v$. We show that voters' criticalities can also be computed by a pseudo-polynomial algorithm.}  
Note that, under our individual uniformity assumption, the probability of delegating for each $i\in V_d$ is equal as they all have the same out-degree. 
We denote this value by $p_d$ and let $p_v = 1-p_d$. 
Moreover, let $n_v = |V_v|$ and $n_d = |V_d| = n-n_v$.

We provide more compact formulas for the LD Penrose-Banzhaf measure in this PV setting. We only consider binary voting rules as there cannot be delegation cycles in this setting ($T^0 = \emptyset$).   

To measure how critical an agent $i$ can be, we consider partitions of $V\backslash\{i\}$ into three sets $V^+$, $V^-$, $V^i$ where $V^+$ (resp. $V^-$) represents the $n^+$ (resp. $n^-$) voters whose final vote is in favour of (resp. against) the proposal, either by voting directly or indirectly and $V^i$ is the set of $n^i$ voters who delegate to voter $i$. 
Note that $V^+$, $V^-$, $V^i$ form a partition of $V\backslash \{i\}$ and $V^i=\emptyset$ when $i\in V_d$.
We focus on how these sets intersect $V_d$ and $V_v$.  
We define $V_d^+$, $V_d^-$, $V_v^+$, and $V_v^-$ with size $n^+_d$, $n^-_d$, $n^+_v$, and $n^-_v$, respectively, such that  $V_x^\circ = V_x \cap V^\circ$ for $x\in \{v,d\}$ and $\circ\in\{-,+\}$.

Given our probabilistic model of delegation partitions, observe that the probability of having a partition $V^+, V^-, V^i$ only depends on these cardinalities. More precisely: 

\noindent -- When $i \in V_v$, note that $n^-_v = n_v - 1 - n^+_v $ and $n^i = n_d - n_d^+-n_d^-$.
Hence, we denote this probability of having such a partition $V^+, V^-, V^i$ by $P_{v}(n^+_v,n^+_d, n^-_d)$: 
\begin{align}
P_{v}(n^+_v,n^+_d, n^-_d) =& \frac{1}{2^{n_v -1}} (\frac{p_v}{2} \!+\! p_d\frac{n_v^+}{n_v})^{n_d^+} \notag\\
&\times(\frac{p_v}{2} \!+\! p_d\frac{n_v^-}{n_v})^{n_d^-}(\frac{p_d}{n_v})^{n^i}. \label{eq:pdtee beta} 
\end{align}
-- When $i \in V_d$, note that $n_v^- = n_v - n_v^+$ and $n_d^- = n_d-1-n_d^+$. We let  $P_d(n_v^+,n_d^+)$ denote the probability of having such a partition of $V^+, V^-$: 
\begin{align}
P_{d}(n^+_v,n^+_d) =& \frac{1}{2^{n_v}} (\frac{p_v}{2} \!+\! p_d\frac{n_v^+}{n_v})^{n_d^+}(\frac{p_v}{2} \!+\! p_d\frac{n_v^-}{n_v})^{n_d^-}. \label{eq:pdtor beta}
\end{align} 
There are some conditions on the integer parameters $n_v^+$, $n_d^+$, and $n_d^-$. 
If $i \in V_v$, we have that $n^+_v \le n_v-1$, and $n^+_d + n^-_d \le n_d$.  If $i \in V_d$, we have $n^+_v \le n_v$, and $n^+_d + n^-_d = n_d - 1$. 
If these conditions are not respected, we set $P_{v}(n^+_v,n^+_d, n^-_d) = 0$ (resp. $P_{d}(n^+_v,n^+_d) = 0$).

We now detail Equation~\ref{eq:pdtee beta}. Equation~\ref{eq:pdtor beta} is obtained similarly.  
The probability of the binary votes of the delegatees other than $i$ being a certain way is $(\nicefrac{1}{2})^{n_v-1}$. 
Then, the probability that each voter in $V^+_d$ (resp. $V^-_d$) votes in favour of the proposal is $\nicefrac{p_v}{2}+ p_d\nicefrac{n_v^+}{n_v}$ (resp. against is $\nicefrac{p_v}{2} + p_d\nicefrac{n_v^-}{n_v}$) where the first summand corresponds to the case in which the voter votes and the second to the one in which they delegate.  
Last, the probability that each voter in $V^i_d$ delegates to $i$ is $\nicefrac{p_d}{n_v}$. 
Equation~\ref{eq:pdtee beta} is obtained by taking the products of these terms.

Given the probability of having a partition $V^+$, $V^-$, $V^i$ of $V\setminus\{i\}$, the voting power measure for a voter  in our PV setting $i\in V$ can be formulated in the following way. 

\begin{proposition} \label{def : PMdtee}
Given a bipartite digraph $G=(V,E)$ with $V=(V_d,V_v)$ and $E = \{(i,j) | i\in V_d, j \in V_v\}$ and a binary voting rule $W$, the \emph{LD Penrose-Banzhaf measure} $\M{ld}_i(W,G)$ of voter $i\in V$ can be formulated as:
\begin{align*}
\M{ld}_i(W,G) &= \!\!\!\!\!\!\! \sum_{\substack{ V^+_v, V^-_v \\ \in \mathcal{P}_2(V_v \setminus\{i\})}} \sum_{\substack{ V^+_d, V^-_d, V^i \\ \in \mathcal{P}_3(V_d)}} P_{v}(n^+_v,n^+_d, n^-_d)\\
&~~~~~\times \delta_{i,-\rightarrow +}^{W}(V^+, V^-,V^i) \text{ if } i \in V_v.\\
\M{ld}_i(W,G) &= \!\!\!\! \sum_{\substack{ V^+_v, V^-_v \\ \in \mathcal{P}_2(V_v)}} \sum_{\substack{ V^+_d, V^-_d \\ \in \mathcal{P}_2(V_d \setminus\{i\})}} P_{d}(n^+_v,n^+_d)\\
&~~~~~\times \delta_{i,-\rightarrow +}^{W}(V^+, V^-,\emptyset) \text{ if } i \in V_d.
\end{align*}
\end{proposition}

We return to Example~\ref{ex:introunderlying} to illustrate our power measures. Notably, we shall see that a voter in $V_v$ with a small weight can achieve a higher criticality through delegation. 

\begin{example}\label{ex:ProxyIndex}
Consider the voters in Example~\ref{ex:introunderlying}; however, now in the PV setting, we assume that $V_v=\{a,b,c\}$ and  $V_d=V\backslash V_v$ and we compute the voters' LD Penrose-Banzhaf measures using Proposition \ref{def : PMdtee}. 
The resulting power measures can be seen in Table~\ref{tab:EXpowerindicies} when $p_d=0,0.5,0.9$. 
As those in $V_d$ have the possibility of voting directly as well as delegating, they have more influence on the outcome  when they are more likely to vote directly; conversely, those in $V_v$ have less as they are less likely to receive delegations. 
When $p_d=0$, all agents vote and thus, we return to a standard WVG with the standard Banzhaf measure where all voters with the same weight have the same voting power.

\begin{table}[]
    \centering
    \caption{$\M{ld}$ when $p_d=0,0.5, 0.9$ for voters $V=\{a,\cdots, m\}$ in the PV setting with $V_v = \{a,b,c\}$ (Values are rounded to  3 d.p.). }
    \label{tab:EXpowerindicies}
    \begin{tabular}{|c|c|c|c|c|}
    \hline
       Agent $x\in V$  &  $p_d=0$ & $p_d=0.5$ & $p_d=0.9$ \\
         \hline
   $a$: $w=3$  & $0.511$& $ 0.552$& $0.542$\\
    $b$: $w=2$ & $0.306$& $ 0.395$ & $0.438$\\
    $c$: $w=1$  &$0.148$& $0.303$ & $0.390$\\
    \hline
    $d$: $w=2$  & $0.306$& $0.206$ &$0.138$\\
     $V_d\backslash\{d\}$: $w=1$ & $0.148$& $0.098$ &$0.065$\\
    \hline
    \end{tabular}
\end{table}
\end{example}


\paragraph{Computational aspects}  
We turn to some computational aspects regarding the PV setting. 
We obtain that the exact computation of the LD measure of voting power is $\#$P-hard due to Proposition~\ref{prop : extend BP}~\citep{prasad1990np}.   
More positively, we show that in WVGs, when restricting the underlying graph to represent the  PV setting that, $\M{ld}$ can be computed in pseudo-polynomial time, similarly to the Penrose-Banzhaf measure. 
This result relies on the following lemma.

\begin{restatable}{lemma}{lemmaPVpseudo}\label{lemma : PV pseudo}
Given a WVG with weight function $w$ and an integer $c$. Computing the number of ways of having a partition $(S_1,S_2,\ldots,S_c)$  in $\mathcal{P}_c(S)$ of a set $S\subseteq V$ with sizes $n_1$, $n_2$,$\ldots$, $n_c = |S| - \sum_{l=1}^{c-1}n_l$, and weights $w(S_1)= w_1$, $w(S_2) = w_2$, $\ldots$, and $w(S_c) = w_c =w(S) - \sum_{l=1}^{c-1} w_l$ can be computed in pseudo-polynomial time.  
\end{restatable}

\begin{restatable}{theorem}{theoremPVPseudo}
Given a bipartite digraph $G=(V,E)$ with $V=(V_d,V_v)$ and $E = \{(i,j) | i\in V_d, j \in V_v\}$,  a WVG with weight function $w$ and quota-ratio $q$, and a voter $i$, measure $\M{ld}_i$ can be computed in pseudo-polynomial time.
\end{restatable}

\section{Liquid democracy with complete digraph}\label{sec:complete} 
This section discusses the case where $G = (V,E)$ is complete, representing LD where  any voter can vote directly, or delegate their vote to any other voter. 
Since the graph is complete, every voter has the same out-degree $|V|-1$. Under our individual uniformity assumption, this implies that the probability to delegate $p_d$ is the same for every voter.
As with PV, we provide a more compact formulation of our power measure by grouping over similar voters instead of summing over all delegation partitions.  By abuse of notation, we say that a set $S$ of voters form an in-forest when the graph obtained by having a vertex per voter in $S$ and an arc from $i$ to $j$ when $i$ delegates to $j$ forms an in-forest. 
We consider a partition of $V\setminus\{i\}$ into four sets $V^+$, $V^-$, $V^0$, $V^i$ where $V^+$ (resp. $V^-$) is a set of $n^+$ (resp. $n^-$) voters voting directly in favour of (resp. against) the issue or indirectly by transitively delegating to a root voter in $V^+$ (resp. $V^-$); $V^0$ is a set of $n^0$ voters abstaining as their delegation leads to a delegation cycle; and $V^i$ is the set of $n^i$ voters delegating (directly or not) to $i$.  
Note that $V^+$,$V^-$,$V^0$ and $V^i$ form a partition of $V\setminus \{i\}$.

We will use recursive formulas to compute the probability of having such a partition into four sets. 
Let $P^{ld}(m,p)$ be the probability that $m$ voters in a set $S\subseteq V$ form an in-forest where the roots all make the same action\footnote{ $P^{ld}(m,p)$ depends only on $|S|$ and $p$, and not on the list of voters in $S$, and thus is independent n the choice of voters in $S$. }; an action which is chosen by each root voter with probability $p$. 
For instance, $P^{ld}(n^+,p_v/2)$ would be the probability that the voters in $V^+$ form a forest where each root voter is in favour of the proposal. 
Consider an arbitrary voter $j\in S$, and a two partition $(S_1,S_2) \in \mathcal{P}_2(S\setminus\{j\})$ with respectively $m_1$ and $m_2=m-1-m_1$ voters. 
The voters in $S_1$ are those who delegate directly or indirectly to $j$, while voters in $S_2$ do not. 
Another way of seeing it is that all voters in $S_1$ form an in-forest where every root delegates to voter $j$ (with probability $\nicefrac{p_d}{(n-1)}$), while voters in $S_2$ form an in-forest where every root realizes the same action as in $S$. 
Regarding voter $j$, there are two possibilities: either voter $j$ realizes the same action as the roots of $S$ (e.g., voting for the proposal), or they delegate to a member of $S_2$ (with probability $\nicefrac{p_dm_2}{(n-1)}$). 

Hence, we obtain the following recursive formula:
\begin{align}
    P^{ld}(m,p) = \sum_{m_1=0}^{m-1} &{m-1 \choose m_1}P^{ld}(m_1,\frac{p_d}{n-1})P^{ld}(m_2,p) \notag\\
    &\times(p + p_d\dfrac{m_2}{n-1}) \label{eq : rec ld}
\end{align}
with base case $P^{ld}(1,p) = p$ and $P^{ld}(0,p) = 1$. 

Thus, the probability that $V^+$ (resp. $V^-$) forms an in-forest where the roots vote in favour of (resp. against) the issue is $P^{ld}(n^+,p_v/2)$ (resp. $P^{ld}(n^-,p_v/2)$); and that the probability that $V^i$ forms an in-forest where the roots delegate to voter $i$ is $P^{ld}(n^i,p_d/(n-1))$. 
For $V^0$, we need a different formula. 
Voters in $V^0$ have their delegation leading to a delegation cycle through other voters in $V^0$ iff each voter in $V^0$ delegates to another voter in $V^0$. 
This occurs with probability $P_0^{ld}(n^0) = \left(\nicefrac{ p_d(n^0-1)}{(n-1)}\right)^{n^0}$. 

To sum up, the probability of having a four partition $(V^+,V^-,V^i, V^0)$ of $V\setminus\{i\}$ is $P^{ld}(n^+,p_v/2) 
P^{ld}(n^-,p_v/2)
P^{ld}(n^i,p_d/(n-1))
P_0^{ld}(n^0)$.

\begin{proposition} \label{def : Mdtee}
Given a complete digraph $G=(V,E)$ and a ternary voting rule $W$, the \emph{LD Penrose-Banzhaf measure} $\M{ld}_i(W,G)$ of voter $i\in V$ can be formulated as:
\begin{align*}
\M{ld}_i(W) = \!\!\!\!\!\!\! 
\sum_{\substack{V^+, V^-, V^0, V^i\\ 
\in \mathcal{P}_4(V \setminus\{i\})}}& P^{ld}(n^+,\frac{p_v}{2})P^{ld}(n^-,\frac{p_v}{2})\\
\times P^{ld}(n^i,\frac{p_d}{n-1})& P_0^{ld}(n^0)\delta_{i,-\rightarrow +}^{W}(V^+, V^-, V^i).
\end{align*}
\end{proposition}

\begin{example}\label{ex:LDpowerindices}
We return to the agents $V=\{a, \cdots, m\}$ from the previous examples, with the same weights as before;  however, as we are in the LD setting where the underlying network is a complete digraph. In Table~\ref{tab:EXpowerindiciesLD}, we see the power measures of each agent where the probability of delegating varies. 
 When $p_d=0$, we are in the standard weighted voting game case where all agents vote directly. When $p_d=0.5$, those with less voting weight have their voting power measure increase, this is due to the possibility of others delegating to them and the voting weight they control becoming higher. Observe that when $p_d=1$, all agents are caught in delegation cycles and $T^0=V$. Thus, we study when $p_d=0.9$. 

\begin{table}[]
    \centering
     \caption{\small $\M{ld}_x$ (rounded to 3 d.p.) for $p_d\in \{0,0.5,0.9\}$ for $v=\{a,\cdots, m\}$ from Example~\ref{ex:introunderlying} when considering a complete network. }
    \label{tab:EXpowerindiciesLD}
    \begin{tabular}{|c|c|c|c|}
    \hline
       Agent $x\in V$  & $p_d=0$  & $p_d=0.5$  & $p_d=0.9$   \\
         \hline
    $a$: $w=3$ &  $0.511$ &  $0.424$  & $0.696$\\
    $b,d$: $w=2$ & $0.306$ & $0.308$ & $0.638$\\
    $V\backslash\{a,b,d\}$: $w=1$ & $0.148$ & $0.212$ & $0.568$ \\
    \hline
    \end{tabular}
\end{table}
\end{example}

In simulated examples, similar to Example~\ref{ex:LDpowerindices}, we noticed two trends. First, a \emph{flattening effect} on the power measures as $p_d$  increased. By this, we mean that the difference between the lowest and highest measure of power in the WVG (for any agent) becomes smaller. For instance, in Table~\ref{tab:EXpowerindiciesLD}, this difference is $0.438$, $0.266$, and $0.103$ for $p_d = 0$, $0.5$, and $0.9$, respectively. This flattening, in our LD setting, is due to all voters having the same available voting actions, no matter their weights. 
Notably, there can be no dummy agents when $p_d >0$, as for any agent, the delegation partition where all other voters delegate to them has a positive probability.
Second, as illustrated by Table~\ref{tab:EXpowerindiciesLD}, we see that when the probability of delegating increases, so does the probability of being critical, especially when the weights are equal.\footnote{We conjecture that when voting weights are equal, the probability of being critical strictly increases with $p_d$.} As when $p_d$ increases, the number of direct voters decreases while the expected accumulated weight of an agent increases. Hence, they are more likely to be critical when they vote directly.
Although it seems intuitive that as the probability of delegating increases, so does the probability of being critical, this is not generally true. In Table~\ref{tab:EXpowerindiciesLD}, we indeed observe that the criticality of voter $a$ decreases when $p_d$ increases from $0$ to $0.5$. 

\paragraph{Computational aspects}
Using Proposition~\ref{def : Mdtee} and Lemma~\ref{lemma : PV pseudo}, we show that, if the digraph is complete, our power measure can be computed in pseudo-polynomial time for WVGs. 

\begin{restatable}{theorem}{theoremLDPseudo}
Given a complete digraph $G=(V,E)$, a WVG with weight function $w$ and quota-ratio $q$, and a voter $i$, $\M{ld}_i$ can be computed in pseudo-polynomial time.
\end{restatable}
The idea of this result is as follows. 
We can compute the number of ways $\lambda$ of having a partition $(S^1,S^2,S^3,S^4)$  in $\mathcal{P}_4(V{\setminus}\{i\})$ with sizes $n^+$, $n^-$, $n^0$, $n^i$, and weights $w^+$, $w^-$, $w^0$, and $w^i$ using Lemma~\ref{lemma : PV pseudo}, and may compute the product $\lambda \times P^{ld}(n^+,\frac{p_v}{2})P^{ld}(n^-,\frac{p_v}{2})P^{ld}(n^i,\frac{p_d}{n-1})P_0^{ld}(n^0)$.  
The result is the sum of these terms for the different tuples $(n^+, n^-, n^0, n^i, w^+, w^-, w^0, w^i)$  for which $i$ is critical.
The number of tuples to be considered is bounded by $n^3\times w(V)^3$.

\section{Experiments}
We performed numerical tests on our power measure to test the impact and relationships between different parameters. 
For each experiment, we estimate the criticality of voters by sampling over delegation-partitions due to the long runtimes required for  exact calculations. 
Details about the sampling and the confidence interval it entails can be found in Appendix~\ref{subsec:omittedexperiments}, as well as additional experimental results. 

We first computed the criticality of voters with a variety of underlying networks.\footnote{
The specifics of the different networks and the figures for these experiments can be found in Appendix \ref{subsec:omittedexperiments}.} 
First, we observed a strong correlation between the voters' criticality and their in-degree in the network. This follows the intuition that the higher the in-degree of a voter, the higher the number of voters that can delegate to them. 
Second, we noticed that the type of the underlying network had a large impact on the differences between the voters' criticality. In particular, inequality in voting power was the largest on  preferential attachment networks \citep{barabasialbert} and the smallest on small-world networks \citep{watts1998collective}. 

In the remainder of this section, we focus on our two special cases, bipartite digraphs  and complete digraphs.

\subsection*{The number of delegators in proxy voting}
In the experiments with proxy voting, we study the case when all voters have the same voting weight and delegators can delegate to any delegatee, as in Section~\ref{sec:bipartite}. 
Note that within either $V_v$ or $V_d$ that all voters have the same voting power. We inspect the effect of $p_d$ in the PV setting, i.e.,  does the probability of those in $V_d$ delegating affect the probability of being critical for both those in $V_v$ and $V_d$. We set $|V|=100$ and look at two different number of delegatees, $|V_v| \in \{20,50\}$.

\begin{figure}[t]
    \centering
    \scalebox{1}{\input{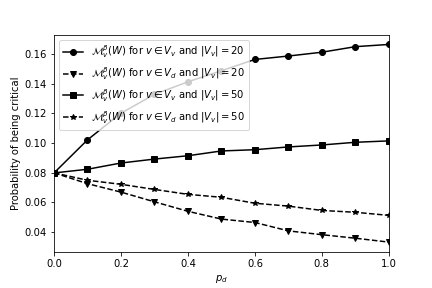}}
    \caption{\small The probability of an agent being critical in the PV setting with $p_d$ varying from $0$ to $1$.
    We have $|V| = 100$, $|V_v| \in\{20,50\}$, and $W$ is a WVG with all weights equal to $1$ and $q=0.5$. This experiment sampled over 100,000 delegations partitions.
    }
    \label{fig:varypd}
\end{figure}

In Figure~\ref{fig:varypd}, when $p_d=0$, we obtain the standard voting model where all agents vote directly and thus have the same chance of being critical. In both instances, as $p_d$ increases, so does the delegatees' probability of being critical, yet the probability of the delegators being critical decreases, reflecting the intuition that there is some transfer of power from the delegators to the delegatees when $p_d$ increases. 
Observe that the difference between the criticality of the delegators and delegatees is smaller when $|V_v|=50$ than when $|V_v|=20$ for every value of $p_d$, as a higher number of delegatees share a lower number of delegators. Thus in the PV setting, increasing $|V_v|$ will flatten the probability of being critical.

\subsection*{The effect of the voters' weights}

We study the impact of $p_d$ in the LD model where the underlying network is complete as in Section~\ref{sec:complete}. 
We have $|V| = 100$ voters, with $50$ voters (resp. $30$ and $20$) having weight $1$ (resp. $2$ and $5$). The quota of the WVG remains $q=0.5$. We vary $p_d$ between $0$ and $0.9$. 
In the case $p_d = 1$, all voters delegate to each other, and thus they all have a criticality of $1$. In Figure~\ref{fig:ldcomplete}, we see that voters with higher weights have higher voting power. 
We observe a flattening effect: the initial gap between the criticality of agents with different weights gets increasingly smaller when $p_d$ increases. As in Table~\ref{tab:EXpowerindiciesLD}, the criticality of voters with smaller weights always increases while it is not the case for voters with weight $5$.
\begin{figure}[t]
    \scalebox{1}{\input{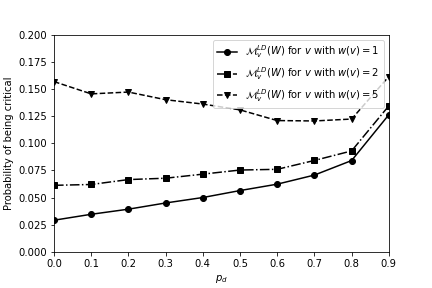}}
	\caption{\small Probability of an agent being critical with $p_d$ varying from $0$ to $0.9$     when the underlying graph is complete. We have $V=100$, and $W$ is a WVG with $50$ (resp. $30$, $20$) voters with weights equal to $1$ (resp. $2$,  $5$) and $q=0.5$. This experiment sampled over 10,000 delegations partitions.
    \label{fig:ldcomplete}}
\end{figure}

\section{Conclusion}

This paper continues the tradition of extending the notion of a priori voting power to new voting models. We have introduced the \emph{LD Penrose-Banzhaf measure} to evaluate how critical voters are in deciding the outcome of an election where delegations play a key role. We study a general setting where an underlying graph restricts the possible delegations of the voters. 
We provided a hardness result on the computation of our measure of voting power. 
Nevertheless, we designed a sampling procedure to estimate them as well as two pseudo-polynomial algorithms that can be used when the graph restricting the delegations is either bipartite or complete.

Several directions are conceivable for future works. First, one could study the same models with more voting options, such as abstention. We have restricted ourselves to two voting options (approving or disapproving) to keep these new models simple. Another direction would be to find the conditions, such as adding or removing neighbours,  that affect the power measure.   Additionally, extending the Coleman indices, one could study how to differentiate in our setups the ability to support an initiative from the one to veto it. 
Lastly, analysing real-election data using our model is a promising option.

\newpage
\appendix

\section*{Acknowledgments}
Rachael Colley acknowledges the support of the ANR JCJC project SCONE (ANR 18-CE23-0009-01). Théo Delemazure was supported by the PRAIRIE 3IA Institute under grant ANR-19-P3IA-0001 (e). Hugo Gilbert acknowledges the support from the project THEMIS ANR-20-CE23-0018 of the French National Research Agency (ANR).


\bibliographystyle{named}
\bibliography{ijcai23}
\balance

\newpage
\appendix
\clearpage
\appendix
\section{Appendix of Paper}

\subsection{Omitted Proofs} \label{subsec:omitted}

\theoremLDMeasure*
\begin{proof}[Proof sketch.]
 $\M{ld}_i(W,G)=\mathbbm{P}(\text{i is critical})$  as it sums the probability of a delegation-partition being critical, determined by $W( T_{D^+_i}) - W( T_{D^-_i})$. 
 Being positively critical means that changing a vote for to against the issue will also change the outcome in the same way (negatively critical is defined similarly). If $W$ is symmetric then $W(T_{D^+_i}) - W( T_{D^-_i})$ captures both $i$ being positively or negatively critical. If $\OutN{i}=\emptyset$, then their only option is to vote either for or against the issue. In both cases, $\mathbbm{P}(\text{i is positively critical}) = \M{\gamma}_i(W)/2$  $=\mathbbm{P}(\text{i is negatively critical})$. 
\end{proof}

\theoremHardness*
\begin{proof} 
To prove this, we will give a reduction from the problem of counting simple paths in a digraph which is known to be \#P-complete \citep{valiant1979complexity}.  
The problem takes as input a directed graph $G = (V,E)$ and nodes $s,t \in V$, the problem then returns the number of simple paths from $s$ to $t$ in $G$.

Let $G = (V,E)$ be a directed graph and $s,t \in V$ two connected vertices of this graph. We want to count simple paths starting at $s$ and finishing at $t$. In order for our problems to align, we will make some alterations to $G$. We remove every incoming edge of $s$ and every outgoing edge of $t$ from $E$ and we add a dummy vertex $z$ to $V$ with no outgoing or incoming edges.

Let $d_{max}$ be the highest out-degree of any vertex in $G$. Let $k \ge 0$ and $G_k = (V_k, E_k)$ be a graph obtained from $G$, such that $V_k=V\cup\{z\}\cup D_k$ with $|D_k|=d_{max}+k$ a set of dummy vertices. For every vertex $x \in V \setminus \{t,z\}$, we add edges from $x$ to $d_{max}+k - \OutN{x}$ dummy vertices from $D_k$ to get $E_k$. By doing this, we ensure that every vertex in $V_k$ either have $d_{max}+k$ or $0$ out-neighbours. Note that we have only added a polynomial number of vertices and edges to $G$ to obtain $G_k$, if $k$ is smaller than some polynomial in $n$. 

Note that  the maximum length of a simple path from $s$ to $t$ in $G$ is less than the total number $n$ of vertices in $G$. 
Now, we will compute the criticality of $t$ in the graph $G_k$ with $k = 0, \dots,n$ for the WVG in which $w(s) = 3|V_k|$, $w(z) = |V_k|$ and $w(x) = 1$ for all $x \in V_k\backslash\{s,z\}$ and $q = \frac{1}{2}$. 

Observe that when $s$ is not in a cycle, then $s$ is in some sense a dictator as $w(s)> q\sum_{x\in V_k}w(x)= \frac{1}{2}(5|V_k|-2)$ and the election's outcome is determined by their vote. However, if $s$ is in a cycle, then the outcome of the election depends only on $z$'s vote, $w(z)>  q\sum_{x\in V_k\backslash\{s\}}w(x)=\frac{1}{2}(2|V_k|-2)$ (furthermore, note that $z$ cannot be in a cycle as they have no outgoing edges). 
Thus, a vertex $x \in V_k\backslash\{s,z\}$ is critical given a delegation partition iff $x$ is part of $s$'s delegation path. In particular, this is true for node $t$.

Let $P$ be a delegation path from $s$ to $t$ in $G_k$. Each vertex on this path has $d_{max}+k$ outgoing edge by definition (otherwise they would have $0$ outgoing edges and could not be on this path). Thus, the probability of obtaining this path of length $\ell$ is equal to $\left ( p_d^{d_{max}+k}\frac{1}{d_{max}+k} \right )^\ell$ under the individual uniformity assumption, where the probability to delegate $p_d^{d_{max}+k}$ is the same for all voters as they all have the same number of outneighbors. We define $\pi_k =p_d^{d_{max}+k}\frac{1}{d_{max}+k}$ and  
let $\mathcal P_l$ be the set of paths of length $\ell$ between $s$ and $t$ in $G$, for $\ell \le n$. We know that $t$ is critical iff there is a path from $s$ to $t$ on the delegation graph, and that there cannot be two different paths from $s$ to $t$ in the same delegation graph. Thus, the criticality of $t$ is equal to the sum of the probabilities of every path from $s$ to $t$ to be selected in the delegation graph:

\begin{align}\label{eq:compproof} 
     \mathbbm{P}(t \text{ is critical in }G_k) = \sum_{\ell= 0}^{n} |\mathcal P_\ell|\pi_k^\ell.
\end{align}

Thus, from Equation~\ref{eq:compproof} we obtain $n+1$ linear equations (one from each election, i.e., each value of $k$) where the variables are $|\mathcal{P}_{\ell}|$. This gives us a set of $n+1$ equations, which is polynomial in the number of vertices. 

The coefficient matrix $M$ of these equations is a Vandermonde matrix with $M_{ij} = \pi_i^j$ for $i,j = 0, \dots, n$ and all the $\pi_i$ values are different. Thus, if we denote the two vectors $X^P = (|P_0|, \dots, |P_{n}|)$ and $Y^c$ such that $Y^c_k = \mathbbm{P}(t \text{ is critical in }G_k)$, we have the equation $Y^c = MX^P$. 
As $M$ is a Vandermonde matrix with different coefficients, it is easily invertible and $X^P$ could be computed in polynomial time. 
Hence, we would obtain the values of $|P_1|, \dots, |P_n|$ and we could easily derive the total number of paths from $s$ to $t$ in $G$ by doing $\sum_{\ell = 0}^{n} |P_{\ell}|$. This would solve the \#P-complete problem of counting simple paths from $s$ to $t$. 
\end{proof}

\lemmaPVpseudo*
\begin{proof}
We first order (arbitrarily) the voters in $S$ such that $v_i$ is the $i^{th}$ voter in $S$. We will denote by $S[i]$ the subset $\{v_i, v_{i+1}, \ldots{,} v_{|S|}\}$ of $S$.
The proof then relies on the following recursive formula:
\begin{align}
\Lambda[i&,n_1,\ldots,n_{c-1},w_1,\ldots,w_{c-1}] = \notag\\ &\Lambda[i\!+\!1,n_1\!-\!1,n_2,\ldots,n_{c-1},w_1\!-\!w(v_i),w_2,\ldots,w_{c-1}] \label{eq : rec1}\\
&\!+\! \Lambda[i\!+\!1,n_1,n_2\!-\! 1,\ldots,n_{c{-}1},w_1,w_2\!-\!w(v_i),\ldots, w_{c-1}]\label{eq : rec2}\\
&\vdots\notag\\
&\!+\! \Lambda[i\!+\!1,n_1,n_2,\ldots,n_{c-1}\!-\!1,w_1,w_2,\ldots, w_{c-1}\!-\!w(v_i)]\notag\\
&\!+\!\Lambda[i\!+\!1,n_1,n_2,\ldots,n_{c-1},w_1,w_2,\ldots,w_{c-1}]\label{eq : rec4}.
\end{align}
The term $\Lambda[i,n_1,n_2,\ldots,n_{c-1},w_1,w_2,\ldots,w_{c-1}]$ denotes the number of ways of having a group of $n_1$ voters in $S[i]$ with total weight $w_1$, and $n_2$ other voters in $S[i]$ having total weight $w_2$, \ldots, and $n_{c-1}$ other voters in $S[i]$ having total weight $w_{c-1}$.  
The number of ways of having a partition $(S_1,S_2,\ldots,S_c)$  in $\mathcal{P}_c(S)$ with sizes $n_1$, $n_2$, \ldots, $n_c = |S| - \sum_{l=1}^{c-1}n_l$, and weights $w(S_1)= w_1$, $w(S_2) = w_2$, $\ldots$, and $w(S_c) = w_c =w(S) - \sum_{l=1}^{c-1} w_l$ is then $\Lambda[1,n_1,n_2,\ldots,n_{c-1},w_1,w_2,\ldots, w_{c-1}]$.  
Let us now explain the recursive formula: the term on line~\ref{eq : rec1} (resp. line~\ref{eq : rec2},  line~\ref{eq : rec4}) counts the number of such partitions of $S$ when $v_i$ is part of the first group of $n_1$ voters (resp. second group of $n_2$ voters, neither of the $c-1$ first groups, hence the last group). 
The base cases are as follows: $\Lambda[i,n_1,\ldots,n_{c-1},w_1,\ldots,w_{c-1}] = 0$ if at least one of the parameters is inferior to 0; $\Lambda[i,0,0,\ldots,0] = 1$ for $i\in\{1,\ldots, |S|+1\}$; and $\Lambda[|S|+1,n_1,\ldots,n_{c-1},w_1,\ldots,w_{c-1}] = 0$ if at least one of the $2(c-1)$ last parameters is different from 0. 
For a fixed $c$ value, using memoization, this recursive formula can be computed in polynomial time with respect to $n$ and $\max_v w(v)$. 
\end{proof}

\theoremPVPseudo*
\begin{proof}
We give the details of the more complex case when $i\in V_v$.
We consider all the possibilities of having a two-partition of $V_v\setminus\{i\}$ with sets of sizes $n^+_v$ and $n^-_v$ and weights $w^+_v$, and $w^-_v$ in conjunction with a three partition of $V_d$ with sets of sizes $n^+_d$, $n^-_d$, $n^i$, and weights $w^+_d$, $w^-_d$, and $w^i$. 
Such a tuple (with 10 elements) will be called a \emph{decomposition} of $V$ informally.
The number of such decompositions is bounded by $n_v\times n_d^2\times w(V)^3$. 
Given a decomposition, we say that $i$ is critical if $w^+_v + w^+_d$ is in the interval $(q\times w(V) - w(i) - w^i, q\times w(V)]$. 

On the one hand, we compute the number of ways $\lambda_1$ of having a partition $(S^1,S^2,S^3)$  in $\mathcal{P}_3(V_d)$ with sizes $n^+_d$, $n^-_d$, $n^i$, and weights $w^+_d$, $w^-_d$, and $w^i$.  
On the other hand, we compute the number of ways $\lambda_2$ of having a partition $(S^1,S^2)$  in $\mathcal{P}_2(V_v\setminus\{i\})$ with sizes $n^+_v$, $n^-_v$, and weights $w^+_v$, and $w^-_v$. 
Both operations are performed using Lemma~\ref{lemma : PV pseudo} in pseudo-polynomial time. Last, we compute the product $\lambda_1\times \lambda_2 \times P_v(n_v^+, n_d^+, n_d^-)$.  
The result is the sum of these terms for the different possible decompositions for which $i$ is critical. 
\end{proof}

\theoremLDPseudo*
\begin{proof}
Let $w_{\max} = \max_v w(v)$. 
 We consider all the possibilities of having a four partition of $V\setminus\{i\}$ with sets of sizes $n^+$, $n^-$, $n^0$, and $n^i$ and weights $w^+$, $w^-$, $w^0$, $w^i$. 
 Such a tuple (with 8 elements) will be called a \emph{decomposition} of $V$ informally.
 The number of such decompositions is bounded by $n^3\times w(V)^3$. 
 Given a decomposition, we say that $i$ is critical if $w^+$ is in $(q\times (w(V)-w^0) - w(i) - w^i, q\times (w(V)-w^0)]$. 

 Then, we compute the number of ways $\lambda$ of having a partition $(S^1,S^2,S^3,S^4)$  in $\mathcal{P}_4(V{\setminus}\{i\})$ with sizes $n^+$, $n^-$, $n^0$, $n^i$, and weights $w^+$, $w^-$, $w^0$, and $w^i$ using Lemma~\ref{lemma : PV pseudo}. 
 Last, we compute the product $\lambda \times P^{ld}(n^+,\frac{p_v}{2})P^{ld}(n^-,\frac{p_v}{2})P^{ld}(n^i,\frac{p_d}{n-1})P_0^{ld}(n^0)$. 
 The result is the sum of these terms for the different possible decompositions for which $i$ is critical. 
\end{proof}

\subsection{The Delegative Banzhaf measure as a measure of I-power.} \label{subsec:zhangGrossi}
To differentiate our measures from the delegative Banzhaf measure defined by \citet{zhang2021power}, we now provide a probabilistic model on voters' behaviours to rationalize it as a measure of I-power. 

Similarly, as in WVGs, the authors study elections in which each voter has a voting weight, and a coalition wins if and only if its ``accumulated weight'' exceeds the quota. 
Additionally, we are given an in-forest where vertices are voters and an arc from $i$ to $j$ represents that $i$ delegates to $j$. 
Thus, each voter in an in-forest has a direct voter (the corresponding root voter) via a chain of delegations who will vote on their behalf. 
Unlike in standard WVGs, a coalition's accumulated weight is not considered as the sum of its members' weights but only those for which the delegation chain to their direct voter is included in the coalition. Hence, a voter whose chain of delegation requires a voter outside of the coalition to reach their root voter will not contribute to the weight of the coalition. 
The delegative Banzhaf measure is then defined as the standard Banzhaf-Penrose measure in this voting game. 

This measure can be rationalized by the following probabilistic model on voters' behaviours. 
Each voter has a probability $0.5$ of having a positive a priori opinion about the proposal and a probability $0.5$ of having a negative opinion about the proposal (these probabilities being independent); this leads to a uniform probability distribution on bipartitions of the voters set. Given a bipartition of opinions and an in-forest,  we define a direct-voting partition in the following recursive way: each voter votes in favour of the proposal if i) they have a positive opinion about the proposal, and ii) either they vote directly or the person they directly delegated to in the in-forest also votes in favour of the proposal. If one of these conditions is not met, the voter will vote against the proposal. 
Put differently, the in-forest provided in input gives each voter a condition for their support. 
The voting rule used is then the same as in standard WVGs. Given this probabilistic model, the delegative Banzhaf measure is the probability of being critical. 

As one can see, the delegative Banzhaf measure differs from the power measures we define as the input to define them and the intuitions underlying them both differ.

\subsection{Connection between WVGs with randomised weights and our proxy voting settings}\label{appendix:randomWVGs} 

Weighted voting games in our proxy setting, i.e., when the underlying network is a bipartite digraph, can be seen as a special case of weighted voting games with random weights where the weights are obtained through our probabilistic model on delegations. 

However, in this section, we will stress some differences between the settings. 
First, note that the set of voters voting directly is also randomised as members of $V_d$ may or may not vote. Moreover, the random processes generating the weights for the delegatees are specific to our proxy settings. 
These weights are constrained by the weights of the delegators and our probabilistic delegation model. 
\citet{boratyn2020average} claim to be the first to study the Banzhaf index with random weights. 
In their setting, the voters' weights are drawn uniformly from a probability simplex, providing a continuous distribution from which to draw their (normalised) weights. 
This entails that the discrete nature of the delegators' weights cannot be adequately modeled within this framework as there is no way to ensure that the weights drawn correspond to an actual coalition (for example, that some delegator's weight isn't being split across two delegatees). 
Instead of using the uniform distribution, one could try to use the probability distribution on the weight simplex yielded by our proxy voting models. Unfortunately, note that this probability distribution would have a support that at worst is of exponential size. Hence, it seems difficult to express this probability distribution compactly while ignoring the delegation process that generates it. 

In other papers that study the Shapley-Shubik index~\citep{filmus2016shapley,bachrach2016characterization}, the authors consider weighted-voting games where voters' weights are drawn independently from one another under some probability distribution. Once again, because of this independence assumption, these models do not seem to encapsulate our proxy voting settings. 

\subsection{An alternative proxy voting model}\label{appendix:PValpha}
In this section, we explore an alternative proxy voting model in which $V$ is partitioned in two sets $V_v$ of size $n_v$ and $V_d$ of size $n_d = n - n_v$. Similarly to the PV model discussed in our paper, the set $V_v$ contains delegatees who will vote. Differently, voters in $V_d$ can only delegate their votes to proxies in $V_v$. Hence, this assumption restricts the class of admissible delegation partitions. 
We term this restricted setting of proxy voting PV$_r$.  
\begin{definition}
Given a non-empty set $V_v\subseteq V$, a \emph{PV$_r$-partition} $P$ is a map $D$ who takes value in $V$ and for which $D(v) \in \{-1,1\}$ if $v\in V_v$, and $D(v) \in V_v$ otherwise. 
\end{definition}

We now provide an example to illustrate the differences between our two PV settings.
\begin{example} \label{example1App}
\begin{figure}
    \centering
\begin{tikzpicture}[node distance=0.25cm,semithick]

\node at  (0.5,1) {$d_2$};
\node at  (1.5,1) {$e_1$};
\node at  (2.5,1) {$f_1$};
\node at  (3.5,1) {$g_1$};
\node at  (4.5,1) {$h_1^-$};
\node at  (5.5,1) {$i_1^+$};
\node at  (0.5,2) {$a^-_3$};
\node at  (1.5,2) {$b^-_2$};
\node at  (2.5,2) {$c^+_1$};

\node at  (3.6,2) {\small PV$_r$};
\node at  (5.6,2) {\small PV};

\draw[-stealth, thick] 
(0.5, 1.2) edge  (0.5,1.8)
(1.5, 1.2) edge  (1.5,1.8)
(2.5, 1.2) edge  (2.5,1.8)
(3.5, 1.2) edge  (2.52,1.8);

 \draw[rounded corners, dotted] (-0.1,2.2) rectangle (4, 0.6);
 \draw[rounded corners, dashed] (-0.2,2.3) rectangle (6, 0.5);
\end{tikzpicture} 
    \caption{This figure depicts the agents' delegations in Example~\ref{example1App} such that the subscript gives the agent's weight and the superscript of $\boldsymbol{+}$ or $\boldsymbol{-}$ represents the direct vote of the non-delegating agents. When restricting the agents to be those inside the dotted, and dashed lines, we obtain a valid PV$_{r}$, and PV-partition, respectively. }
    \label{fig:introexample2}
\end{figure}
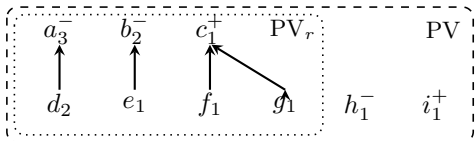

Consider a set of agents $V=\{a,b,\cdots, i\}$ with $V_v=\{a,b,c\}$ such that $D^+=\{c,i\}$, $D^-=\{a,b,h\}$, $D^a=\{d\}$, $D^b=\{e\}$, and $D^c=\{f,g\}$ as described in Figure~\ref{fig:introexample2}. This delegation partition is valid for the PV model described in our paper where $G = (V,E)$ is a bipartite graph, in which $V = (V_d,V_v)$ and $E = \{(i,j)|i\in V_d, i\in V_v\}$. 
However, while this example, restricted to voters $\{a,b,\ldots,g\}$ form a valid \emph{PV$_r$-partition}, it is not the case when considering voters $h$ and $i$ as they vote directly and yet belong to $V_d$.
\end{example}

We now revisit the individual uniformity and global uniformity assumptions under the lens of this model. 

\paragraph{The individual uniformity assumption}
Under the individual uniformity assumptions, each voter in $V_v$ has equal probability of supporting or opposing the proposal. Additionaly, each voter in $V_d$ has equal probability of delegating to any voter in $V_v$. Voters behaviors are independent from one another. These assumptions are justified by the fact that we have no information as to voters' personality and interests, or the nature of the proposal and we are ignorant of any concurrence or opposition of interests between voters. Hence, we assume that $p_y = p_n = \nicefrac{1}{2}$ for voters in $V_v$ and that voters in $V_d$ may delegate to any member of $V_v$ with probability $\nicefrac{1}{n_v}$. 

\paragraph{The global uniformity assumption}
If there is no information about the proposal or voters, we assume that all PV$_r$ partitions are equally likely. 
We call this assumption the \emph{global uniformity assumption}, as it does not rely directly on a model of individual behaviours. 
Hence, the probability of a given partition being chosen is $(\nicefrac{1}{2})^{n_v}\times(\nicefrac{1}{n_v})^{n_d}$. 

\paragraph{Remark.} Note that the \emph{individual} and \emph{global} uniformity assumptions lead to the same probabilistic model under the the PV$_r$ setting.

\paragraph{A measure of voting power for the $PV_r$ setting}
We want to measure how critical a voter is in determining the outcome. Given our probabilistic model on PV$_r$ partitions, we consider the number of times a voter can change the outcome decided by a binary voting rule.
\begin{itemize}
   \item If a voter belongs to $V_v$, then we  look if the agent can impact the outcome by changing their vote, either to be against the proposal (positive criticality) or in favour of the proposal (negative criticality). 
   \item If a voter belongs to $V_d$, then we look if an agent has impact by changing their delegation (if possible), either by switching from a delegatee in favour to one against the proposal (positive criticality) or from a delegatee against to one in favour of the proposal (negative criticality). 
\end{itemize}
Given a PV$_r$-partition, the agent is critical if they are positively or negatively critical. 

To measure how critical an agent $i$ is, we consider a partition of $V\backslash\{i\}$ into three sets $V^+$, $V^-$, $V^i$:
\begin{itemize}
    \item $V^+$ represents the $n^+$ voters whose final vote is in favour of the proposal, either by voting directly or indirectly by delegating to a delegatee in $V_v\backslash\{i\}$.
    \item $V^-$ represents the $n^-$ voters who vote directly against the proposal or indirectly by delegating to a delegatee in $V_v\backslash\{i\}$.
    \item $V^i$ is the set of $n^i$ voters who delegate to voter $i$. This set is empty if $i\in V_d$.  
\end{itemize}
Note that $V^+\cup V^- \cup V^i= V\backslash \{i\}$ and that these partitioning  sets are disjoint. 
We will focus on how these sets intersect $V_d$ and $V_v$. 
We define $V_d^+$, $V_d^-$, $V_v^+$, and $V_v^-$ with size $n^+_d$, $n^-_d$, $n^+_v$, and $n^-_v$, respectively, such that  $V_x^\circ = V_x \cap V^\circ$ for $x\in \{v,d\}$ and $\circ\in\{-,+\}$.  

Given our probabilistic model on proxy partitions, we observe that the probability of having a partition $V^+, V^-, V^i$ only depends on these cardinality values. 
\begin{itemize}
\item If $i \in V_v$, it only depends on $n^+_v$, $n^+_d$, and $n^-_d$, as $n^-_v = n_v - n^+_v -1$ and $n^i = n_d - n_d^+-n_d^-$. Hence, we then denote this probability by $P_{v}^{r}(n^+_v,n^+_d, n^-_d)$. 
\item If $i \in V_d$, it only depends on $n^+_v$, and $n^+_d$ as $n_v^+ + n_v^- = n_v$ and $n_d^- = n_d-1-n_d^+$. Hence, we then denote this probability by $P_d^{r}(n_v^+,n_d^+)$. 
\end{itemize}
The formulas to compute these probabilities are given below: 
\begin{align}
P_{v}^{r}(n^+_v,n^+_d, n^-_d) =&  \frac{(n^+_v)^{n^+_d}(n^-_v)^{n^-_d}}{2^{n_v-1}n_v^{n_d}} \label{eq:pdtee alpha},\\ 
P_d^{r}(n_v^+,n_d^+) =&   \frac{(n^+_v)^{n^+_d}(n^-_v)^{n^-_d}}{2^{n_v}n_v^{n_d-1}} \label{eq:pdtor alpha}. 
\end{align}
Note that there are some obvious conditions on the integer parameters $n_v^+$, $n_d^+$, and $n_d^-$. 
\begin{itemize}
\item If $i\in V_v$, it should be that $n^+_v \le n_v-1$, and $n^+_d + n^-_d \le n_d$. 
Moreover, values $n_d^+$ (resp. $n_d^-$) should be equal to 0 if $n^+_v = 0$ (resp. $n^+_v = n_v - 1$). 
\item If $i\in V_d$, we should now have $n^+_v \le n_v$, and $n^+_d + n^-_d = n_d - 1$. 
Once more, values $n_d^+$ (resp. $n_d^-$) should be equal to 0 if $n^+_v = 0$ (resp. $n^+_v = n_v$).
\end{itemize}

If some of these conditions are not respected, we set that $P_{v}^{r}(n^+_v,n^+_d, n^-_d) = 0$ or $P_d^{r}(n_v^+,n_d^+) = 0$.
Let us now detail Equation~\ref{eq:pdtee alpha} (Equation~\ref{eq:pdtor alpha} is obtained in a similar way). 
The probability of the binary votes of the delegatees other than $i$ being a certain way is $(\nicefrac{1}{2})^{n_v-1}$. 
Then, the probability that voters in $V^+_d$ (resp. $V^-_d$, $V^i$) delegate to the ones in $V^+_v$ (resp. $V^-_v$, $\{i\}$) is $(\nicefrac{n_v^+}{n_v})^{n_d^+}$ (resp. $(\nicefrac{n_v^-}{n_v})^{n_d^-}$, $(\nicefrac{1}{n_v})^{n^i}$ ). 
Equation~\ref{eq:pdtee alpha} is obtained by considering the products of these terms. 

Given the probability of having a partition $V^+$, $V^-$, $V^i$ of $V\setminus\{i\}$, the voting power measure for a voter $i\in V$ can be defined. 

\begin{definition} 
Given a set $V$ of voters, a set $V_v\subseteq V$ of delegatees, and a binary voting rule $W$, the \emph{PV$_r$} Penrose-Banzhaf measure $\M{r}_i(W)$ of voter $i\in V$ is defined as:
\begin{align*}
\M{r}_i(W) &= \!\!\!\!\!\!\! \sum_{\substack{ V^+_v, V^-_v \\ \in \mathcal{P}_2(V_v \setminus\{i\})}} \sum_{\substack{ V^+_d, V^-_d, V^i \\ \in \mathcal{P}_3(V_d)}} P_{v}^{r}(n^+_v,n^+_d, n^-_d)\\
&\delta_{i,-\rightarrow +}^{W}(V^+, V^-,V^i)  \text{ if } i \in V_v,\\
\M{r}_i(W) &= \!\!\!\!\!\!\! \sum_{\substack{ V^+_v{\neq} \emptyset, V^-_v{\neq} \emptyset \\ \in \mathcal{P}_2(V_v)}} \sum_{\substack{ V^+_d, V^-_d \\ \in \mathcal{P}_2(V_d \setminus\{i\})}} P_{d}^{r}(n^+_v,n^+_d)\\
&\delta_{i,-\rightarrow +}^{W}(V^+, V^-,\emptyset)\text{ if } i \in V_d.
\end{align*}
\end{definition}
Note that when $i \in V_d$, we consider only partitions for which $V^+_v\neq \emptyset$ and $V^-_v\neq \emptyset$. 
Indeed, those in $V_d$ are only able to delegate to delegatees. 
Hence, their ability to be critical depends on the fact that two delegatees with opposite votes exist. Note that the probability that all delegatees vote in the same way is $(\nicefrac{1}{2})^{n_v-1}$.

$\M{r}_i$ quantifies the probability to sample a PV$_r$-partition where $i$ is able to alter the election's outcome. Additionally, observe that our voting power measures extend the standard Penrose-Banzhaf measure (consider $V_v = V$) and that they are not normalized (i.e., summing over the agents does not yield 1). The corresponding voting power indices can be found by normalizing over voters.

\paragraph{Computational aspects}  
We now turn to some computational aspects regarding the PV$_{r}$ measure of voting power. 
While the exact computation of this measure is $\#$P-hard as it extends the standard Banzhaf measure, we show that in WVGs, measures $\M{r}$ can be computed in pseudo-polynomial time. 
This result relies on Lemma~\ref{lemma : PV pseudo}.

\begin{theorem}
Given a WVG with weight function $w$ and quota-ratio $q$, a set of voters $V_v\subseteq V$, and a voter $i$, measure $\M{r}_i$ can be computed in pseudo-polynomial time.
\end{theorem}
\begin{proof}
We give the details of the more complex case when $i\in V_v$.
We consider all the possibilities of having a two partition of $V_v\setminus\{i\}$ with sets of sizes $n^+_v$ and $n^-_v$ and weights $w^+_v$, and $w^-_v$ in conjunction with a three partition of $V_d$ with sets of sizes $n^+_d$, $n^-_d$, $n^i$, and weights $w^+_d$, $w^-_d$, and $w^i$. 
Such a tuple (with 10 elements) will be called a \emph{decomposition} of $V$ informally.
The number of such decompositions is bounded by $n_v\times n_d^2\times w(V)^3$. 
Given a decomposition, we say that $i$ is critical if $w^+_v + w^+_d$ is in the interval $(q\times w(V) - w(i) - w^i, q\times w(V)]$. 

On the one hand, we compute the number of ways $\lambda_1$ of having a partition $(S^1,S^2,S^3)$  in $\mathcal{P}_3(V_d)$ with sizes $n^+_d$, $n^-_d$, $n^i$, and weights $w^+_d$, $w^-_d$, and $w^i$. 
On the other hand, we compute the number of ways $\lambda_2$ of having a partition $(S^1,S^2)$  in $\mathcal{P}_2(V_v\setminus\{i\})$ with sizes $n^+_v$, $n^-_v$, and weights $w^+_v$, and $w^-_v$. 
Both operations are performed using Lemma~\ref{lemma : PV pseudo} in pseudo-polynomial time. Last, we compute the product $\lambda_1\times \lambda_2 \times P_v^{r}(n_v^+, n_d^+, n_d^-)$. 
The result is the sum of these terms for the different possible decompositions for which $i$ is critical. 
\end{proof}

\subsection{Additional details on the experiments} \label{subsec:omittedexperiments}
\paragraph{Sampling procedure.} We show that we can estimate the values $\M{ld}_i(W,G)$ by a sampling procedure which samples ``enough'' $G$-delegation partitions according to the individual uniformity assumption and counting the number of samples for which $i$ is critical. 
To determine how many $G$-delegation partitions should be sampled, we rely on the following well known inequality~\citep{hoeffding1994probability}. 

\begin{theorem}[Hoeffding's inequality] Let $X_1, \ldots, X_m$ be independent random variables, where all $X_i$ are bounded such that $X_i \in [a_i , b_i ]$, and let $X = \sum_{i=1}^m X_i$ . Then the following inequality holds.
$$\mathbbm{P}(|X-\mathbbm{E}[X]|\ge n\epsilon) \le 2\exp \left (-\frac{2m^2\epsilon^2}{\sum_{i=1}^m (b_i-a_i)^2} \right)$$
\end{theorem}

Using this inequality, we can prove the following result.

\begin{theorem}[Estimation by sampling]
Given a graph $G=(V,E)$, a ternary voting rule $W$, values $\epsilon > 0$, and $\delta > 0$, let $\tilde{\mathcal{D}}$ be a set of $k \ge ln(2n/\delta)/(2\epsilon^2)$ $G$-delegation partitions sampled independently and uniformly. Then, 
\begin{align*}
\tilde{\M{ld}}_i(W,G) =  \frac{1}{k} \sum_{D\in \tilde{\mathcal{D}}} \frac{W( T_{D^+_i}) - W( T_{D^-_i})}{2}.
\end{align*}
belongs to $[\M{ld}_i(W,G)-\epsilon,\M{ld}_i(W,G)+\epsilon]$ with probability $1-\delta/n$. 
By using a union bound, the result holds for all voters with probability $1-\delta$.
\end{theorem}
\begin{proof}
Note that $\frac{W( T_{D^+_i}) - W( T_{D^-_i})}{2}$ is a random variable taking value in $\{0,1\}$. Hence, using Hoeffding's inequality, the fact that $k \ge ln(2n/\delta)/(2\epsilon^2)$ and that $\M{ld}_i(W,G) = \mathbbm{E}(\frac{W( T_{D^+_i}) - W( T_{D^-_i})}{2})$, we obtain that for any $i\in V$, $$\mathbbm{P}(|\M{ld}_i(W,G) - \tilde{\M{ld}}_i(W,G) |\ge \epsilon) \le \delta/n.$$ 

For the next step, we use a union bound.
\begin{proposition}[Union bound] Let $E_1$, $\ldots$, $E_m$ be $m$ events, then:
$\mathbbm{P}( \bigcup_{i=1}^m E_i) \leq \sum_{i=1}^m \mathbbm{P}(E_i).$
\end{proposition}
Hence,
\begin{align*} 
&\mathbbm{P}(\exists i \in V \text{ s.t. } |\M{ld}_i(W,G) - \tilde{\M{ld}}_i(W,G) |\ge \epsilon) \le \\
&~~~\sum_{i\in V} \mathbbm{P}(|\M{ld}_i(W,G) - \tilde{\M{ld}}_i(W,G) |\ge \epsilon) \le \delta
\end{align*}
\end{proof}

We first ran experiments to compute the criticality of voters in different types of network, to see if we could find correlations between the criticality and other metrics.

The kind of networks we considered for this first experiments were the following:
\begin{itemize}
    \item \textbf{Random graph $G(n,p)$} with $n$ nodes and each edge having a probability $p$ to be added to the network.
    \item \textbf{Preferential attachment model} based on the model from \citep{barabasialbert} parameterized by the number of nodes $n$ and the number of edges $m$ to add for each new node. For a network of $n$ nodes, we sequentially add the nodes to the graph and when we add node we also attach it to $m$ of the existing nodes, with a probability proportional to their current degree in the graph. More formally, let $V$ be the set of existing nodes when we add a new node $x \notin V$, then the probability to add the edge $(x,v)$ for some $v \in V$ is proportional to the current degree of $v$. This follows the ``rich get richer'' idea as nodes that are already important in the network are more likely to become even more important.
    \item \textbf{Small world model} based on the model from \citep{watts1998collective}. On the contrary to the previous model, the idea is here that no nodes accumulate too many connections, by mainly having edges with its neighbourhood. More formally, for a network of $n$ nodes, we associate each node to a position on a ring and we connect every node to its $k$ nearest neighbours, and every edge is rewired with some probability $p$, given as a parameter. If an edge is rewired for one node, it can be rewired to any other node of the network uniformly at random. In our experiments, we used $p=0.2$
\end{itemize}
Note that in our experiments, we considered the versions of these three models that gave undirected graphs, so every edge is bidirectional. The next two models give us directed graphs.
\begin{itemize}
    \item In the \textbf{spatial model}, every node is first embedded in a 2 dimensional Euclidean space according to some distribution of positions $\mathcal D$, and for each node, we add a directed edge to its $k$ nearest neighbours on this Euclidean space. Note that this graph is directed as one voter node be another node nearest neighbours while the converse is false. We considered to distribution of positions: (i) the Uniform distribution in $[-1,1]^2$ and the Gaussian distribution centered in $0$ with standard deviation $1$. We expect that with the Gaussian distribution, nodes that are embedded close to the center of the plane $(0,0)$ get an advantage against the nodes that have less central positions.
    \item \textbf{$k$-layers model} are graph with $k$ layers, each layer having the same number of nodes, and every node from the layer $j < k$ has a directed edge towards every node from the layer $j+1$.
\end{itemize}

We constructed $5$ networks of each family with $n = 100$ voters and set up the parameters $m$, $k$ or $p$ such that the average in-degree of the graph is $10$. For each network, we ran $5,000$ delegations-partitions. The figures are averaged over the $5$ networks for each type of networks.

Figure~\ref{fig:distrib_criticality} shows the distribution of criticality for every type of graph, by ordering the voters from the most critical to the least critical. We observe that for some networks, there are voters that are far more critical than others. This is the case for the preferential attachment model in which the most critical voter has $\mathcal M_i^{ld}(W) = 0.35$ and the least critical voter has $0.063$. Whereas, voters in networks based on the small world model seem to have almost all the same criticality, as the most critical voter has $M_i^{ld}(W) = 0.159$ and the least critical voter has $0.094$. We note that in spatial networks, the differences of criticality are greater with the Gaussian distribution than with uniform distribution. In particular, we observe that some voters seems to have very little power in deciding the outcome of the vote. Finally, for the $k$-layer experiment (with $k=10$), it is clear that voters from each layer have (almost) the same criticality and that the criticality gets higher with the layer. This can also be seen in Figure 7.b. In particular, voters on the last layer (with no outgoing edges, so they have to vote), are the one with the highest criticality in average. However, one can note that if the criticality increases quickly between the first layers, it will then increase more slowly between the last ones.

We also looked at other metrics of the networks we generated. In particular, we noticed that the criticalities seemed to be very correlated to the degree. Figure~\ref{fig:distrib_degrees} shows the distribution of in-degrees in the graph, from the node with the highest in-degree, to the node with the lowest one. We observe that the curves of this figures seems very similar to the ones from Figure~\ref{fig:distrib_criticality}. 

To observe more clearly the correlations between the in-degree and the criticality, we plotted for each kind of network, the nodes by their in-degree and their criticality. The results are shown in Figure~7. We see the large correlation between the two metrics, especially in the non-directed networks ($G(n,p)$, preferential attachment and small world). However, in the spatial graph, we see some variations between nodes having the same in-degree. This makes sense, as in these graph the positions of the in-neighbours is maybe as important as their number. In the $k$-layer model, every node has in-degree either $10$ or $0$ (for the first layer), see instead of the degree, we look at how critical the nodes of each layer are.

\begin{figure}[t]
  \input{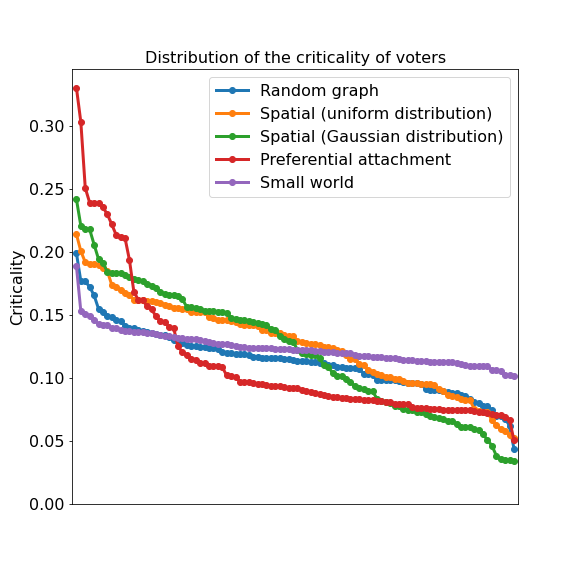}
	\caption{Distribution of the criticality among the voters in the network, from the most critical to the least critical (thus the percentage on the x-axis represents the percentage of voters considered).
    \label{fig:distrib_criticality}}
\end{figure}
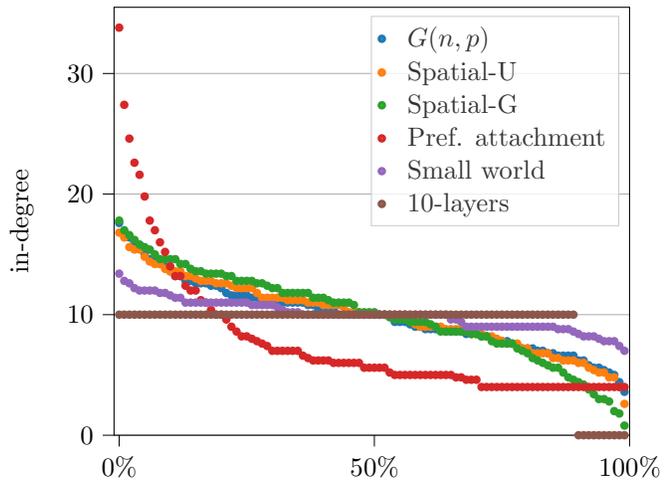
\begin{figure}[t]
  \input{Figures/distribution_degrees}
	\caption{Distribution of the in-degrees among the voters in the network, from the voter with the highest in-degree to the one with the lowest in-degree.
    \label{fig:distrib_degrees}}
\end{figure}

\clearpage

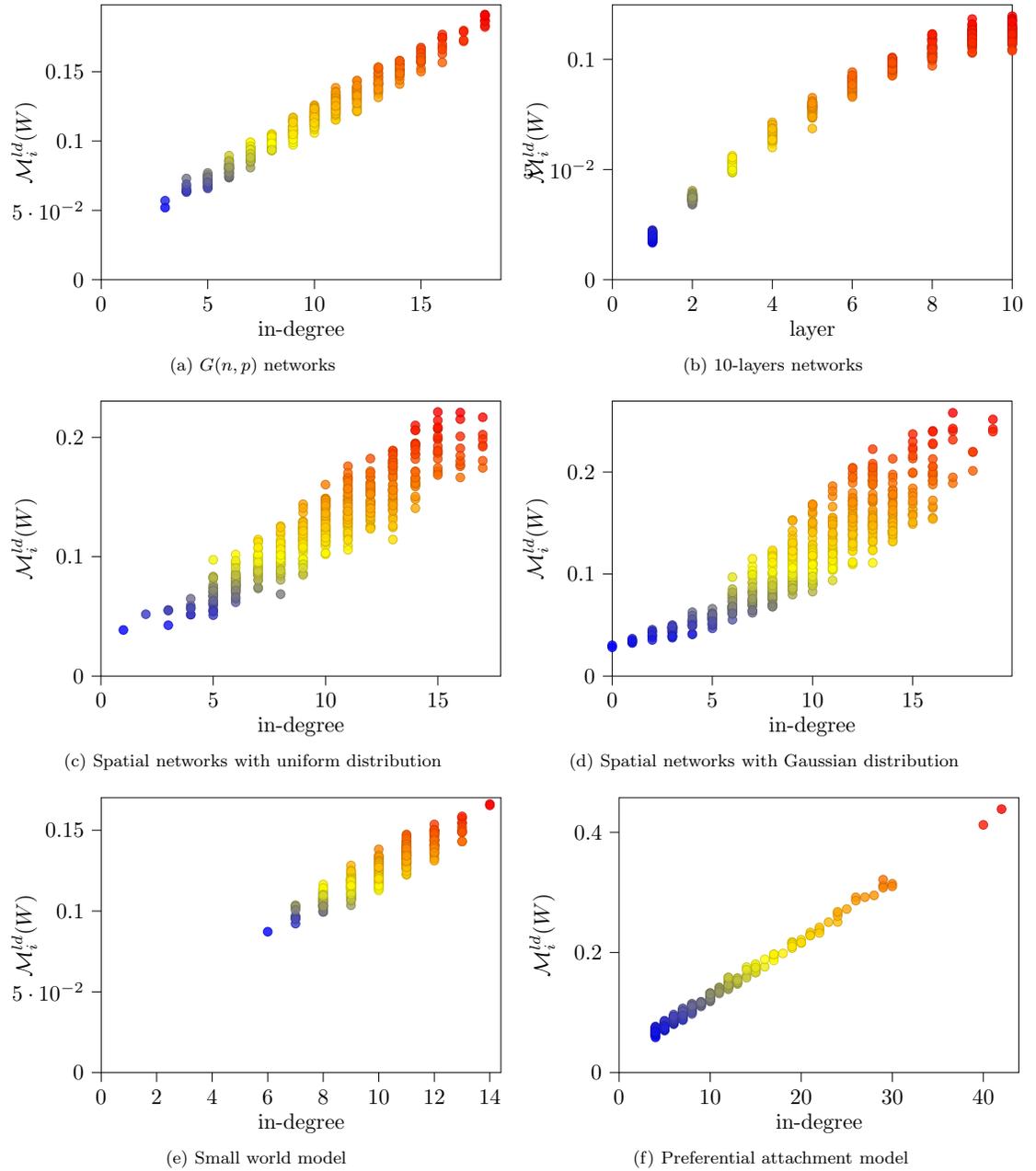
\begin{figure}[!h]\label{fig:correlations}
    \centering
    \scalebox{0.9}{\subfloat[\centering $G(n,p)$ networks]{{
\input{Figures/degree_edges_Gnm}}}}%
    \scalebox{0.9}{\subfloat[\centering $10$-layers networks]{{
\input{Figures/layer_critic_layers}}}}%
    
    \scalebox{0.9}{\subfloat[\centering Spatial networks with uniform distribution]{{
\input{Figures/degree_edges_spatial_uniform}}}}%
    \scalebox{0.9}{\subfloat[\centering Spatial networks with Gaussian distribution]{{
\input{Figures/degree_edges_spatial_gaussian}}}}%

    \scalebox{0.9}{\subfloat[\centering Small world model]{{
\input{Figures/degree_edges_small_world}}}%
    \subfloat[\centering Preferential attachment model]{{
\input{Figures/degree_edges_pref}}}}%
    \caption{ \centering  The correlation between the criticality and the in-degree (or the layer depth for figure (b)) in the different networks considered.}%
\end{figure}

\clearpage

\paragraph{Comparing the two models of proxy voting}

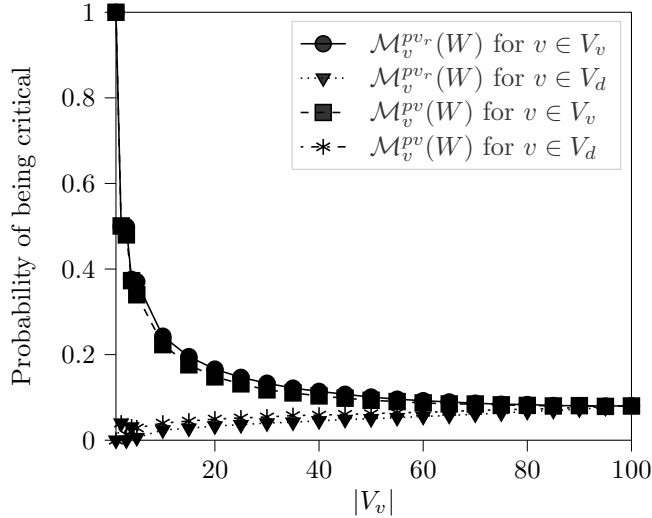
\begin{figure}[h!]
    \centering
    \input{Figures/criticality_pv.tex}
    \caption{Probability of an agent being critical with $|V_v|$ varying from $1$ to $100$ in the standard proxy voting setting $pv$ or the restricted version from Appendix~\ref{appendix:PValpha}, $pv_r$. We have $|V| = 100$ and $p_d = 0.5$ in the $pv$ setting, and $W$ is a WVG with all weights equal to $1$ and $q=0.5$. This experiment sampled over 100,000 delegations partitions.}
    \label{fig:comparingProxyvoting}
\end{figure}

In the experiments with proxy voting ($PV$ as described in Section~\ref{sec:bipartite} and the restricted version $PV_r$ described in Section~\ref{appendix:PValpha}), we study the case when all agents have the same voting weight. Note that within either $V_v$ or $V_d$ that all agents have the same voting power. 
We performed tests on how the number of proxies affects the probability of agents being critical. We study the setting where $|V|=100$ and $|V_v| \in [1,2,3,4,5,10,\dots,95,100]$
(thus, $|V_d|=|V|-|V_v|$), the voting rule corresponds to the majority rule ($q=0.5$), and for the $PV$ setting we have that the probability of delegating is $p_d=0.5$. We see the results of this numerical test in Figure~\ref{fig:comparingProxyvoting}. We see similar trends between the $PV_r$ and $PV$ settings, with the main difference in behaviour being determined by whether an agent is a delegator or a delegatee. 
Yet, the difference in probability of being critical for all voters is slightly less in the $PV$ setting than in the $PV_r$ setting. Returning to the comparisons between the criticality of the delegators and the delegatees, we see that when there are very few proxies in $V_v$, then their likelihood of being critical is very high and approaches $1$. Therefore, their probability of being critical is high due to them receiving many delegations given that $p_d=0.5$ and the choice of delegatee is small. As $|V_v|$ increases, so does the choice of delegatees for those in $V_d$; hence, the probability of being critical for those in $V_v$ decreases. Conversely, the criticality of delegators in $V_d$ slightly raises with this increase in delegation options.

\end{document}

%% file: Figures/criticality_pv_pd_2.tex
\begin{tikzpicture}
\begin{axis}[
legend cell align={left},
legend style={
  fill opacity=0.8,
  draw opacity=1,
  text opacity=1,
  at={(0.03,0.97)},
  anchor=north west,
  draw=white!80!black,
  font=\tiny,
  legend columns=2,
},
tick align=outside,
tick pos=left,
x grid style={white!69.0196078431373!black},
xlabel={\(\displaystyle p_d\)},
xmin=0, xmax=1,
xtick style={color=black},
y grid style={white!69.0196078431373!black},
ylabel={$\mathcal{M}^{ld}_i(W)$},
ymin=0, ymax=0.2,
ytick style={color=black},
ytick={0,0.05,0.1,0.15,0.2, 0.25},
yticklabels={0,0.05,0.1,0.15,0.2, 0.25},
width=8cm, height=6cm
]
\addplot [semithick, black, mark=*, mark size=3, mark options={solid}]
table {%
0 0.0797500000000001
0.1 0.102091999999999
0.2 0.120091999999995
0.3 0.132814999999992
0.4 0.141289999999989
0.5 0.148777999999988
0.6 0.156399499999987
0.7 0.158699499999989
0.8 0.161276999999988
0.9 0.165026499999988
1 0.16652999999999
};
\addlegendentry{
$i \in V_v$ and $|V_v| = 20$}
\addplot [semithick, black, dashed, mark=triangle*, mark size=3, mark options={solid,rotate=180}]
table {%
0 0.0798504999999995
0.1 0.0725233749999995
0.2 0.066819374999999
0.3 0.060272999999999
0.4 0.0538598749999992
0.5 0.048609249999999
0.6 0.0462999999999991
0.7 0.0406483749999996
0.8 0.0380703749999996
0.9 0.0356959999999996
1 0.0330548749999996
};
\addlegendentry{
$i \in V_d$ and $|V_v| = 20$}
\addplot [semithick, black, mark=square*, mark size=3, mark options={solid}]
table {%
0 0.0797714000000005
0.1 0.0822498000000148
0.2 0.0864656000000174
0.3 0.0891208000000205
0.4 0.0913792000000215
0.5 0.0946040000000209
0.6 0.0954518000000203
0.7 0.0973218000000249
0.8 0.0986534000000266
0.9 0.100462400000026
1 0.101405600000024
};
\addlegendentry{
$i \in V_v$ and $|V_v| = 50$}
\addplot [semithick, black, dashed, mark=asterisk, mark size=3, mark options={solid}]
table {%
0 0.0798460000000009
0.1 0.0748192000000008
0.2 0.0720862000000006
0.3 0.0686316000000008
0.4 0.0653080000000008
0.5 0.063323600000001
0.6 0.0593060000000005
0.7 0.0573572000000005
0.8 0.0544430000000005
0.9 0.0531920000000002
1 0.0511182000000001
};
\addlegendentry{
$i \in V_d$ and $|V_v| = 50$}
\end{axis}

\end{tikzpicture}

%% file: Figures/criticality_ld.tex
\begin{tikzpicture}

\begin{axis}[
legend cell align={left},
legend style={fill opacity=0.8, draw opacity=1, text opacity=1, draw=white!80!black,
  font=\tiny,},
tick align=outside,
tick pos=left,
x grid style={white!69.0196078431373!black},
xlabel={\(\displaystyle p_d\)},
xmin=0, xmax=0.9,
xtick style={color=black},
y grid style={white!69.0196078431373!black},
ylabel={$\mathcal{M}^{ld}_i(W)$},
ymin=0, ymax=0.2,
ytick style={color=black},
ytick={0,0.05,0.1,0.15,0.2, 0.25},
yticklabels={0,0.05,0.1,0.15,0.2, 0.25},
width=8cm, height=6cm
]
\addplot [semithick, black, mark=*, mark size=3, mark options={solid}]
table {%
0 0.02923
0.1 0.03469
0.2 0.039384
0.3 0.045096
0.4 0.050104
0.5 0.056514
0.6 0.062394
0.7 0.070756
0.8 0.084112
0.9 0.126202
};
\addlegendentry{
$i$ with $w(i)=1$}
\addplot [semithick, black, dash pattern=on 1pt off 3pt on 3pt off 3pt, mark=square*, mark size=3, mark options={solid}]
table {%
0 0.0613266666666667
0.1 0.0621
0.2 0.06661
0.3 0.06787
0.4 0.0716033333333333
0.5 0.0753933333333333
0.6 0.0760533333333333
0.7 0.0842
0.8 0.09297
0.9 0.13442
};
\addlegendentry{
$i$ with $w(i)=2$}
\addplot [semithick, black, dashed, mark=triangle*, mark size=3, mark options={solid,rotate=180}]
table {%
0 0.15674
0.1 0.14544
0.2 0.14712
0.3 0.14008
0.4 0.13608
0.5 0.13072
0.6 0.120895
0.7 0.12053
0.8 0.122295
0.9 0.1614
};
\addlegendentry{
$i$ with $w(i)=5$}
\end{axis}

\end{tikzpicture}

%% file: Figures/distribution_criticality.tex
\begin{tikzpicture}

\definecolor{color0}{rgb}{0.12156862745098,0.466666666666667,0.705882352941177}
\definecolor{color1}{rgb}{1,0.498039215686275,0.0549019607843137}
\definecolor{color2}{rgb}{0.172549019607843,0.627450980392157,0.172549019607843}
\definecolor{color3}{rgb}{0.83921568627451,0.152941176470588,0.156862745098039}
\definecolor{color4}{rgb}{0.580392156862745,0.403921568627451,0.741176470588235}
\definecolor{color5}{rgb}{0.549019607843137,0.337254901960784,0.294117647058824}

\begin{axis}[
legend cell align={left},
legend style={fill opacity=0.8, draw opacity=1, text opacity=1, draw=white!80!black},
tick align=outside,
tick pos=left,
x grid style={white!69.0196078431373!black},
xmin=-1, xmax=100,
xtick style={color=black},
xtick={0,50,100},
xticklabels={0\%,50\%,100\%},
y grid style={white!69.0196078431373!black},
ylabel={$\mathcal M_i^{ld}(W)$},
ymajorgrids,
ymin=0, ymax=0.3717,
ytick style={color=black}
]
\addplot [very thick, color0, mark=*, mark size=1, mark options={solid}, only marks]
table {%
0 0.1856
1 0.17952
2 0.17452
3 0.16836
4 0.16572
5 0.15996
6 0.15672
7 0.15488
8 0.15252
9 0.15084
10 0.14804
11 0.14696
12 0.14576
13 0.14372
14 0.14184
15 0.14072
16 0.13976
17 0.1384
18 0.13788
19 0.13652
20 0.13496
21 0.13364
22 0.133
23 0.13224
24 0.13144
25 0.1308
26 0.12976
27 0.12864
28 0.12772
29 0.12744
30 0.12708
31 0.1264
32 0.12548
33 0.12492
34 0.12424
35 0.12384
36 0.12348
37 0.123
38 0.12268
39 0.12172
40 0.12072
41 0.12012
42 0.11964
43 0.119
44 0.11836
45 0.11808
46 0.11724
47 0.11692
48 0.11588
49 0.11504
50 0.11452
51 0.114
52 0.11352
53 0.11288
54 0.11196
55 0.11152
56 0.111
57 0.10972
58 0.10928
59 0.10852
60 0.10676
61 0.1064
62 0.10568
63 0.10516
64 0.10488
65 0.10428
66 0.10356
67 0.10276
68 0.10228
69 0.10136
70 0.1004
71 0.09964
72 0.0988
73 0.09832
74 0.0978
75 0.09732
76 0.09624
77 0.09532
78 0.09408
79 0.093
80 0.09276
81 0.092
82 0.09092
83 0.08992
84 0.089
85 0.08836
86 0.0878
87 0.0864
88 0.08456
89 0.08352
90 0.08232
91 0.08112
92 0.08024
93 0.07912
94 0.0766
95 0.07536
96 0.07248
97 0.07064
98 0.06788
99 0.06092
};
\addlegendentry{$G(n,p)$}
\addplot [very thick, color1, mark=*, mark size=1, mark options={solid}, only marks]
table {%
0 0.21072
1 0.20508
2 0.2024
3 0.19868
4 0.19468
5 0.189
6 0.18328
7 0.18044
8 0.1778
9 0.17708
10 0.17544
11 0.17416
12 0.17276
13 0.17164
14 0.17
15 0.16828
16 0.16672
17 0.16456
18 0.16168
19 0.15972
20 0.158
21 0.15728
22 0.15468
23 0.15148
24 0.15084
25 0.15008
26 0.14892
27 0.14804
28 0.14736
29 0.1462
30 0.14484
31 0.14336
32 0.14208
33 0.141
34 0.14052
35 0.13964
36 0.13932
37 0.13852
38 0.1376
39 0.13624
40 0.13564
41 0.1338
42 0.13296
43 0.13228
44 0.13136
45 0.13024
46 0.12808
47 0.12664
48 0.1258
49 0.1252
50 0.12448
51 0.1236
52 0.12224
53 0.12152
54 0.12056
55 0.12004
56 0.1184
57 0.11776
58 0.1172
59 0.1164
60 0.11516
61 0.11436
62 0.11364
63 0.11284
64 0.11128
65 0.11048
66 0.10912
67 0.10808
68 0.1074
69 0.10652
70 0.1048
71 0.10408
72 0.10208
73 0.10116
74 0.09972
75 0.09916
76 0.09836
77 0.09736
78 0.09632
79 0.09496
80 0.09216
81 0.09148
82 0.09032
83 0.08948
84 0.08856
85 0.08588
86 0.08456
87 0.08124
88 0.0802
89 0.0786
90 0.07676
91 0.07552
92 0.073
93 0.0714
94 0.06924
95 0.06716
96 0.06384
97 0.06088
98 0.05856
99 0.04908
};
\addlegendentry{Spatial-U}
\addplot [very thick, color2, mark=*, mark size=1, mark options={solid}, only marks]
table {%
0 0.244
1 0.23228
2 0.22124
3 0.21476
4 0.20796
5 0.20408
6 0.2012
7 0.19972
8 0.19588
9 0.1948
10 0.19048
11 0.18792
12 0.18312
13 0.1812
14 0.17564
15 0.17376
16 0.17232
17 0.1688
18 0.16756
19 0.16408
20 0.16324
21 0.16144
22 0.16036
23 0.15824
24 0.15784
25 0.15704
26 0.15616
27 0.15472
28 0.15336
29 0.152
30 0.15112
31 0.14936
32 0.14864
33 0.147
34 0.1454
35 0.14356
36 0.14292
37 0.1408
38 0.1392
39 0.13812
40 0.13608
41 0.1338
42 0.13136
43 0.13036
44 0.12952
45 0.12728
46 0.12612
47 0.1238
48 0.12168
49 0.12076
50 0.12016
51 0.1176
52 0.11524
53 0.1142
54 0.11352
55 0.11272
56 0.112
57 0.111
58 0.10792
59 0.10588
60 0.10468
61 0.10372
62 0.10308
63 0.10152
64 0.10108
65 0.10064
66 0.09924
67 0.0966
68 0.0952
69 0.093
70 0.0906
71 0.08952
72 0.08808
73 0.0876
74 0.08528
75 0.08156
76 0.08044
77 0.07752
78 0.07656
79 0.074
80 0.07296
81 0.07028
82 0.0684
83 0.06688
84 0.06544
85 0.06332
86 0.06188
87 0.05956
88 0.05692
89 0.05592
90 0.05296
91 0.05156
92 0.04872
93 0.0472
94 0.04528
95 0.04416
96 0.04252
97 0.0406
98 0.0388
99 0.03456
};
\addlegendentry{Spatial-G}
\addplot [very thick, color3, mark=*, mark size=1, mark options={solid}, only marks]
table {%
0 0.35484
1 0.29512
2 0.26988
3 0.24756
4 0.23252
5 0.21788
6 0.20076
7 0.19396
8 0.18476
9 0.17632
10 0.16632
11 0.15956
12 0.15548
13 0.15084
14 0.14628
15 0.14268
16 0.13888
17 0.137
18 0.13192
19 0.12796
20 0.12324
21 0.11904
22 0.11556
23 0.11264
24 0.10952
25 0.108
26 0.10648
27 0.10488
28 0.10308
29 0.10228
30 0.10072
31 0.10012
32 0.09852
33 0.0978
34 0.09608
35 0.09448
36 0.09368
37 0.09204
38 0.09132
39 0.0908
40 0.0906
41 0.08976
42 0.08944
43 0.0882
44 0.08708
45 0.08668
46 0.08608
47 0.08564
48 0.08468
49 0.084
50 0.0834
51 0.08296
52 0.0824
53 0.08196
54 0.08156
55 0.08092
56 0.08008
57 0.07916
58 0.0784
59 0.07788
60 0.0776
61 0.07724
62 0.07664
63 0.07572
64 0.0752
65 0.07496
66 0.07456
67 0.07416
68 0.07372
69 0.07304
70 0.07252
71 0.07212
72 0.07136
73 0.07088
74 0.07056
75 0.07028
76 0.06996
77 0.0694
78 0.06912
79 0.06884
80 0.06868
81 0.06832
82 0.06812
83 0.06784
84 0.06776
85 0.0674
86 0.06728
87 0.06724
88 0.06684
89 0.06656
90 0.06596
91 0.0658
92 0.0654
93 0.06516
94 0.06488
95 0.0648
96 0.06448
97 0.06408
98 0.06352
99 0.06256
};
\addlegendentry{Pref. attachment}
\addplot [very thick, color4, mark=*, mark size=1, mark options={solid}, only marks]
table {%
0 0.15892
1 0.15236
2 0.14968
3 0.14712
4 0.1462
5 0.14548
6 0.14448
7 0.14332
8 0.1424
9 0.14144
10 0.14032
11 0.1396
12 0.13896
13 0.13828
14 0.13716
15 0.13632
16 0.13584
17 0.13524
18 0.13464
19 0.13428
20 0.1338
21 0.13324
22 0.13292
23 0.13256
24 0.13228
25 0.13208
26 0.13156
27 0.131
28 0.13072
29 0.13028
30 0.13016
31 0.1296
32 0.12928
33 0.12908
34 0.12864
35 0.12824
36 0.12816
37 0.12748
38 0.1274
39 0.1268
40 0.12644
41 0.1264
42 0.126
43 0.12572
44 0.12548
45 0.12532
46 0.12516
47 0.12512
48 0.1248
49 0.12452
50 0.12432
51 0.12416
52 0.12408
53 0.12392
54 0.12372
55 0.12352
56 0.1232
57 0.12304
58 0.12252
59 0.1222
60 0.12188
61 0.12144
62 0.12104
63 0.1206
64 0.12008
65 0.11972
66 0.1194
67 0.1192
68 0.11888
69 0.1186
70 0.1184
71 0.1182
72 0.11796
73 0.11712
74 0.11688
75 0.11664
76 0.11628
77 0.116
78 0.11584
79 0.11528
80 0.1146
81 0.11408
82 0.11396
83 0.11376
84 0.11344
85 0.11292
86 0.11264
87 0.11196
88 0.11104
89 0.11056
90 0.11016
91 0.10904
92 0.10836
93 0.10712
94 0.10636
95 0.10504
96 0.10352
97 0.10232
98 0.10016
99 0.09408
};
\addlegendentry{Small world}
\addplot [very thick, color5, mark=*, mark size=1, mark options={solid}, only marks]
table {%
0 0.11668
1 0.1149
2 0.11426
3 0.11398
4 0.11324
5 0.11274
6 0.11226
7 0.11172
8 0.11154
9 0.1111
10 0.1106
11 0.11018
12 0.10974
13 0.1096
14 0.10938
15 0.10888
16 0.10866
17 0.10806
18 0.10774
19 0.10754
20 0.10708
21 0.10666
22 0.10624
23 0.10512
24 0.1044
25 0.10356
26 0.1026
27 0.10214
28 0.10152
29 0.10082
30 0.09982
31 0.0995
32 0.09888
33 0.09834
34 0.09804
35 0.09764
36 0.09624
37 0.09574
38 0.09498
39 0.09374
40 0.09208
41 0.09134
42 0.09046
43 0.08974
44 0.08932
45 0.08854
46 0.08804
47 0.08722
48 0.08566
49 0.08402
50 0.08084
51 0.07984
52 0.0793
53 0.07846
54 0.07788
55 0.07758
56 0.07698
57 0.07644
58 0.07506
59 0.07324
60 0.06948
61 0.06774
62 0.06712
63 0.06644
64 0.06538
65 0.06478
66 0.06402
67 0.0638
68 0.06326
69 0.06238
70 0.05466
71 0.05352
72 0.0531
73 0.05244
74 0.05224
75 0.05178
76 0.05132
77 0.05078
78 0.05014
79 0.04952
80 0.0391
81 0.03844
82 0.03808
83 0.0378
84 0.03724
85 0.03694
86 0.03648
87 0.03622
88 0.03594
89 0.03504
90 0.02146
91 0.02084
92 0.02032
93 0.02016
94 0.02004
95 0.01978
96 0.01952
97 0.01914
98 0.0186
99 0.01764
};
\addlegendentry{$10$-layers}
\end{axis}

\end{tikzpicture}

%% file: Figures/distribution_degrees.tex
\begin{tikzpicture}

\definecolor{color0}{rgb}{0.12156862745098,0.466666666666667,0.705882352941177}
\definecolor{color1}{rgb}{1,0.498039215686275,0.0549019607843137}
\definecolor{color2}{rgb}{0.172549019607843,0.627450980392157,0.172549019607843}
\definecolor{color3}{rgb}{0.83921568627451,0.152941176470588,0.156862745098039}
\definecolor{color4}{rgb}{0.580392156862745,0.403921568627451,0.741176470588235}
\definecolor{color5}{rgb}{0.549019607843137,0.337254901960784,0.294117647058824}

\begin{axis}[
legend cell align={left},
legend style={fill opacity=0.8, draw opacity=1, text opacity=1, draw=white!80!black},
tick align=outside,
tick pos=left,
x grid style={white!69.0196078431373!black},
xmin=-1, xmax=100,
xtick style={color=black},
xtick={0,50,100},
xticklabels={0\%,50\%,100\%},
y grid style={white!69.0196078431373!black},
ylabel={in-degree},
ymajorgrids,
ymin=0, ymax=35.49,
ytick style={color=black}
]
\addplot [very thick, color0, mark=*, mark size=1, mark options={solid}, only marks]
table {%
0 17.6
1 16.8
2 16.4
3 16
4 15.8
5 15
6 14.6
7 14.4
8 14.2
9 14
10 13.6
11 13.2
12 13.2
13 12.8
14 12.8
15 12.6
16 12.6
17 12.6
18 12.4
19 12.4
20 12.2
21 11.8
22 11.6
23 11.6
24 11.6
25 11.6
26 11.4
27 11.4
28 11.2
29 11
30 11
31 11
32 11
33 11
34 11
35 11
36 11
37 10.8
38 10.8
39 10.4
40 10.4
41 10.4
42 10.4
43 10
44 10
45 10
46 10
47 10
48 10
49 10
50 10
51 10
52 10
53 9.8
54 9.4
55 9.4
56 9.4
57 9.2
58 9
59 9
60 8.8
61 8.8
62 8.8
63 8.8
64 8.6
65 8.6
66 8.6
67 8.6
68 8.4
69 8.4
70 8.4
71 8.2
72 8.2
73 8
74 7.8
75 7.8
76 7.8
77 7.6
78 7.4
79 7.2
80 7.2
81 7.2
82 7
83 7
84 6.8
85 6.8
86 6.6
87 6.6
88 6.6
89 6.6
90 6.2
91 6.2
92 5.8
93 5.6
94 5.6
95 5.4
96 5.2
97 5
98 4.4
99 3.6
};
\addlegendentry{$G(n,p)$}
\addplot [very thick, color1, mark=*, mark size=1, mark options={solid}, only marks]
table {%
0 16.8
1 16.4
2 15.6
3 15.4
4 15.4
5 14.8
6 14.4
7 14.2
8 14.2
9 13.8
10 13.6
11 13.6
12 13.6
13 13.2
14 13.2
15 13
16 12.8
17 12.8
18 12.8
19 12.6
20 12.6
21 12.6
22 12.4
23 12.2
24 12.2
25 12.2
26 12.2
27 11.8
28 11.4
29 11.4
30 11.4
31 11.4
32 11.4
33 11.2
34 11.2
35 11.2
36 11.2
37 11
38 11
39 11
40 10.8
41 10.8
42 10.8
43 10.8
44 10.6
45 10.2
46 10.2
47 10.2
48 10
49 10
50 10
51 10
52 10
53 10
54 10
55 10
56 10
57 9.6
58 9.2
59 9
60 9
61 9
62 9
63 9
64 8.8
65 8.8
66 8.8
67 8.8
68 8.6
69 8.4
70 8.4
71 8.4
72 8.2
73 8.2
74 8
75 7.8
76 7.6
77 7.6
78 7.6
79 7.2
80 7.2
81 6.8
82 6.8
83 6.8
84 6.8
85 6.4
86 6.4
87 6.2
88 6.2
89 6.2
90 6
91 6
92 5.6
93 5.4
94 5.2
95 5.2
96 4.8
97 4.8
98 4
99 2.6
};
\addlegendentry{Spatial-U}
\addplot [very thick, color2, mark=*, mark size=1, mark options={solid}, only marks]
table {%
0 17.8
1 17
2 16.6
3 16.2
4 15.8
5 15.6
6 15.4
7 15
8 14.6
9 14.6
10 14.6
11 14.6
12 14.2
13 14.2
14 13.8
15 13.6
16 13.6
17 13.4
18 13.4
19 13.4
20 13.4
21 13.2
22 13.2
23 12.8
24 12.8
25 12.8
26 12.8
27 12.6
28 12.6
29 12.4
30 12.2
31 12.2
32 11.8
33 11.8
34 11.8
35 11.8
36 11.8
37 11.8
38 11.4
39 11.4
40 11.4
41 11.2
42 11
43 11
44 11
45 11
46 10.8
47 10.2
48 10.2
49 10.2
50 10.2
51 10
52 10
53 9.8
54 9.6
55 9.6
56 9.6
57 9.4
58 9.4
59 9.4
60 9.4
61 9.2
62 9
63 8.6
64 8.6
65 8.6
66 8.6
67 8.6
68 8.6
69 8.4
70 8.4
71 8.2
72 8.2
73 7.8
74 7.6
75 7.6
76 7.6
77 7.6
78 7.2
79 7
80 6.8
81 6.4
82 6.2
83 6
84 5.8
85 5.6
86 5.6
87 5.2
88 4.8
89 4.6
90 4.4
91 4.2
92 3.8
93 3.4
94 3
95 3
96 2.8
97 2
98 1.8
99 0.8
};
\addlegendentry{Spatial-G}
\addplot [very thick, color3, mark=*, mark size=1, mark options={solid}, only marks]
table {%
0 33.8
1 27.4
2 24.6
3 22.6
4 21.6
5 19.8
6 17.8
7 17
8 16
9 15.2
10 14
11 13.2
12 13.2
13 12.4
14 12
15 12
16 11.2
17 11
18 10.4
19 10
20 10
21 9.6
22 9
23 8.6
24 8.2
25 8.2
26 8
27 7.8
28 7.6
29 7.4
30 7
31 7
32 7
33 7
34 7
35 7
36 6.6
37 6.4
38 6.2
39 6.2
40 6.2
41 6.2
42 6
43 6
44 6
45 6
46 6
47 6
48 5.6
49 5.6
50 5.6
51 5.6
52 5.6
53 5.2
54 5
55 5
56 5
57 5
58 5
59 5
60 5
61 5
62 5
63 5
64 5
65 5
66 4.8
67 4.8
68 4.6
69 4.6
70 4.6
71 4
72 4
73 4
74 4
75 4
76 4
77 4
78 4
79 4
80 4
81 4
82 4
83 4
84 4
85 4
86 4
87 4
88 4
89 4
90 4
91 4
92 4
93 4
94 4
95 4
96 4
97 4
98 4
99 4
};
\addlegendentry{Pref. attachment}
\addplot [very thick, color4, mark=*, mark size=1, mark options={solid}, only marks]
table {%
0 13.4
1 12.8
2 12.6
3 12.2
4 12
5 12
6 12
7 12
8 11.8
9 11.8
10 11.6
11 11.4
12 11.4
13 11
14 11
15 11
16 11
17 11
18 11
19 11
20 11
21 11
22 11
23 11
24 11
25 11
26 10.8
27 10.8
28 10.8
29 10.8
30 10.8
31 10.6
32 10.4
33 10.2
34 10.2
35 10.2
36 10
37 10
38 10
39 10
40 10
41 10
42 10
43 10
44 10
45 10
46 10
47 10
48 10
49 10
50 10
51 10
52 10
53 10
54 10
55 10
56 10
57 10
58 10
59 10
60 10
61 10
62 10
63 10
64 10
65 9.6
66 9.6
67 9.4
68 9
69 9
70 9
71 9
72 9
73 9
74 9
75 9
76 9
77 9
78 9
79 9
80 9
81 9
82 9
83 9
84 9
85 9
86 8.8
87 8.8
88 8.8
89 8.6
90 8.4
91 8.2
92 8.2
93 8.2
94 8
95 7.8
96 7.8
97 7.8
98 7.4
99 7
};
\addlegendentry{Small world}
\addplot [very thick, color5, mark=*, mark size=1, mark options={solid}, only marks]
table {%
0 10
1 10
2 10
3 10
4 10
5 10
6 10
7 10
8 10
9 10
10 10
11 10
12 10
13 10
14 10
15 10
16 10
17 10
18 10
19 10
20 10
21 10
22 10
23 10
24 10
25 10
26 10
27 10
28 10
29 10
30 10
31 10
32 10
33 10
34 10
35 10
36 10
37 10
38 10
39 10
40 10
41 10
42 10
43 10
44 10
45 10
46 10
47 10
48 10
49 10
50 10
51 10
52 10
53 10
54 10
55 10
56 10
57 10
58 10
59 10
60 10
61 10
62 10
63 10
64 10
65 10
66 10
67 10
68 10
69 10
70 10
71 10
72 10
73 10
74 10
75 10
76 10
77 10
78 10
79 10
80 10
81 10
82 10
83 10
84 10
85 10
86 10
87 10
88 10
89 10
90 0
91 0
92 0
93 0
94 0
95 0
96 0
97 0
98 0
99 0
};
\addlegendentry{$10$-layers}
\end{axis}

\end{tikzpicture}

%% file: Figures/degree_edges_Gnm.tex
\begin{tikzpicture}

\definecolor{color0}{rgb}{0.12156862745098,0.466666666666667,0.705882352941177}

\begin{axis}[
height=6cm,
tick align=outside,
tick pos=left,
width=8cm,
x grid style={white!69.0196078431373!black},
xlabel={in-degree},
xmin=0, xmax=18.75,
xtick style={color=black},
y grid style={white!69.0196078431373!black},
ylabel={\(\displaystyle \mathcal{M}_i^{ld}(W)\)},
ymin=0, ymax=0.19816,
ytick style={color=black}
]
\addplot [draw=color0, fill=color0, mark=*, only marks, opacity=0.8, scatter]
table{%
x  y
5 0.0706
10 0.1148
5 0.0728
6 0.0738
9 0.1092
7 0.0874
9 0.1006
10 0.114
13 0.1438
16 0.1676
13 0.1396
9 0.1078
8 0.095
10 0.1114
12 0.136
12 0.1356
13 0.1414
9 0.1112
6 0.0736
9 0.1056
11 0.1248
12 0.1356
10 0.1144
14 0.1496
9 0.0972
9 0.1056
10 0.1104
11 0.1256
6 0.0832
14 0.1512
14 0.1508
8 0.0944
12 0.1318
7 0.0898
11 0.117
8 0.0932
6 0.0744
5 0.0704
10 0.121
10 0.111
13 0.141
12 0.1212
6 0.0774
8 0.0962
6 0.0796
12 0.1244
10 0.11
10 0.1202
9 0.11
7 0.0876
11 0.1294
6 0.079
9 0.1044
5 0.0668
11 0.1232
8 0.1036
15 0.1522
15 0.1644
9 0.1016
12 0.1388
14 0.1512
10 0.1142
16 0.1628
14 0.1508
8 0.1048
12 0.1224
5 0.0754
7 0.0808
9 0.1026
11 0.122
10 0.1152
16 0.1746
10 0.1126
8 0.0992
10 0.1146
4 0.0646
12 0.1268
9 0.107
3 0.052
13 0.1344
9 0.1078
8 0.098
18 0.1822
7 0.0868
9 0.1046
13 0.1434
11 0.1154
11 0.1258
10 0.109
7 0.0862
6 0.081
8 0.0998
14 0.148
17 0.1796
10 0.1178
13 0.1382
8 0.0972
11 0.1232
11 0.1246
6 0.0804
8 0.1034
9 0.1066
6 0.0842
10 0.1186
9 0.1022
6 0.083
12 0.1362
14 0.1578
11 0.1262
8 0.0946
14 0.1554
9 0.1136
11 0.1268
11 0.131
11 0.1384
4 0.0636
7 0.0884
7 0.0948
9 0.1148
11 0.1318
5 0.0742
15 0.1606
15 0.161
11 0.1254
13 0.1476
12 0.1338
11 0.1332
10 0.1174
9 0.1086
8 0.0998
10 0.1258
9 0.115
6 0.0804
9 0.1068
9 0.1056
11 0.1254
13 0.1458
10 0.1204
11 0.1264
14 0.1512
12 0.1278
10 0.112
8 0.1036
16 0.1768
10 0.1148
7 0.0946
14 0.1578
10 0.1168
10 0.1136
4 0.0672
9 0.1082
7 0.0886
11 0.1294
11 0.1318
10 0.122
10 0.1214
9 0.1084
8 0.1028
15 0.1674
12 0.1342
7 0.0964
12 0.1378
11 0.1312
13 0.1448
9 0.1104
9 0.1102
10 0.1226
6 0.0802
10 0.121
7 0.0936
10 0.1162
12 0.13
9 0.1104
14 0.1554
9 0.106
10 0.1244
15 0.1662
11 0.1248
10 0.1232
11 0.1276
10 0.1208
12 0.1388
11 0.1356
9 0.108
10 0.1058
8 0.0982
11 0.1298
13 0.1424
13 0.1358
10 0.1174
4 0.073
11 0.1318
10 0.1226
7 0.0938
7 0.0862
8 0.1044
8 0.0976
7 0.0874
7 0.0946
8 0.1016
5 0.0746
6 0.0796
8 0.0954
11 0.1268
11 0.1214
9 0.1048
7 0.0856
13 0.1516
17 0.1718
9 0.106
18 0.1834
8 0.1006
4 0.0632
13 0.1414
6 0.0858
8 0.0934
15 0.1582
10 0.1168
11 0.1272
7 0.092
14 0.1482
9 0.11
6 0.0804
10 0.117
11 0.1218
11 0.1234
15 0.1568
11 0.1196
17 0.1728
8 0.0974
14 0.1518
10 0.1156
12 0.1366
8 0.099
7 0.0904
13 0.1366
7 0.0904
7 0.0882
11 0.1246
14 0.141
13 0.1314
5 0.0688
14 0.1462
14 0.144
18 0.1872
13 0.1486
7 0.0848
10 0.1102
7 0.087
12 0.1314
9 0.104
5 0.067
13 0.1464
8 0.0956
9 0.103
8 0.0944
8 0.1038
7 0.0912
8 0.0968
8 0.0982
11 0.1268
10 0.1118
6 0.077
6 0.0764
6 0.0764
7 0.0912
11 0.1152
10 0.1126
10 0.1084
7 0.0914
8 0.0936
7 0.0918
13 0.141
7 0.0888
10 0.1196
5 0.0684
12 0.1222
10 0.1186
12 0.1332
11 0.1288
9 0.1072
11 0.1264
18 0.1862
12 0.1302
4 0.0648
11 0.1208
10 0.112
5 0.0756
16 0.1566
12 0.1216
8 0.1006
11 0.1186
12 0.1306
9 0.107
8 0.1028
11 0.1228
8 0.1006
5 0.0658
13 0.1406
11 0.1248
8 0.0952
16 0.1744
12 0.1432
10 0.1208
17 0.1786
13 0.1462
11 0.1258
12 0.135
6 0.0826
8 0.0966
7 0.099
11 0.122
14 0.1554
8 0.0998
10 0.1148
7 0.0948
8 0.1002
10 0.125
8 0.0988
12 0.1362
4 0.0688
11 0.1298
7 0.09
8 0.1038
5 0.0726
8 0.097
9 0.117
9 0.1136
13 0.153
13 0.149
16 0.1738
10 0.121
8 0.105
9 0.111
7 0.0916
15 0.1632
9 0.1152
14 0.1534
10 0.116
10 0.1184
7 0.0912
8 0.1006
10 0.1148
7 0.0904
10 0.1116
9 0.1164
5 0.077
11 0.1326
7 0.094
5 0.0726
10 0.1192
7 0.0872
6 0.0862
6 0.0858
11 0.1306
13 0.1422
15 0.1602
10 0.117
8 0.1038
18 0.19
10 0.1148
13 0.1494
8 0.1024
9 0.1048
12 0.1348
5 0.0738
6 0.08
7 0.0946
10 0.1206
11 0.1266
13 0.1522
8 0.1026
9 0.105
13 0.1532
18 0.1912
10 0.1232
12 0.1382
10 0.118
6 0.0894
8 0.1026
10 0.1236
15 0.1574
10 0.1104
11 0.125
12 0.136
7 0.099
9 0.1142
11 0.1268
7 0.0934
9 0.106
11 0.135
7 0.0958
11 0.1298
7 0.0948
10 0.1132
12 0.135
11 0.1304
15 0.15
12 0.1436
9 0.1034
11 0.1252
11 0.1266
11 0.1276
9 0.1088
9 0.1056
13 0.1374
14 0.1538
6 0.0818
11 0.131
9 0.1054
11 0.1216
11 0.1212
7 0.0856
8 0.0984
10 0.1156
12 0.13
12 0.1362
11 0.1246
9 0.1072
12 0.1298
10 0.115
11 0.1252
12 0.1336
12 0.1334
8 0.1
5 0.0702
12 0.1382
7 0.0962
16 0.1744
9 0.099
10 0.112
6 0.0824
11 0.1238
7 0.094
14 0.1472
7 0.0908
8 0.103
7 0.0902
11 0.1198
11 0.128
11 0.1236
9 0.1072
7 0.0938
10 0.117
8 0.0936
9 0.1108
16 0.1698
7 0.0854
11 0.1294
11 0.1232
10 0.1238
10 0.11
8 0.0932
12 0.1316
6 0.0818
7 0.0892
11 0.1254
6 0.0806
6 0.075
11 0.1234
7 0.0888
8 0.0982
7 0.0908
9 0.101
18 0.1906
12 0.1328
10 0.1176
8 0.0984
15 0.1566
7 0.088
16 0.166
12 0.1272
11 0.1198
10 0.1184
6 0.0852
9 0.1078
9 0.1082
10 0.1166
11 0.1216
11 0.1288
6 0.077
8 0.0936
9 0.1068
11 0.1256
9 0.1094
11 0.1286
9 0.1048
10 0.1174
10 0.1198
6 0.081
10 0.1172
11 0.1252
10 0.1232
10 0.113
14 0.1578
12 0.14
15 0.158
9 0.1134
3 0.057
14 0.1488
};
\end{axis}

\end{tikzpicture}

%% file: Figures/layer_critic_layers.tex
\begin{tikzpicture}

\definecolor{color0}{rgb}{0.12156862745098,0.466666666666667,0.705882352941177}

\begin{axis}[
height=6cm,
tick align=outside,
tick pos=left,
width=8cm,
x grid style={white!69.0196078431373!black},
xlabel={layer},
xmin=0, xmax=10,
xtick style={color=black},
y grid style={white!69.0196078431373!black},
ylabel={\(\displaystyle \mathcal{M}_i^{ld}(W)\)},
ymin=0, ymax=0.12474,
ytick style={color=black}
]
\addplot [draw=color0, fill=color0, mark=*, only marks, opacity=0.8, scatter]
table{%
x  y
1 0.0208
1 0.0225
1 0.0212
1 0.0212
1 0.0195
1 0.0222
1 0.0179
1 0.0212
1 0.0207
1 0.0206
2 0.037
2 0.0344
2 0.039
2 0.0372
2 0.0397
2 0.0358
2 0.0362
2 0.039
2 0.0386
2 0.0375
3 0.0508
3 0.053
3 0.051
3 0.0522
3 0.0535
3 0.0528
3 0.0562
3 0.0516
3 0.0499
3 0.0531
4 0.0633
4 0.0643
4 0.0665
4 0.0683
4 0.0629
4 0.0643
4 0.0667
4 0.0677
4 0.0684
4 0.0714
5 0.0807
5 0.0805
5 0.0812
5 0.0778
5 0.0757
5 0.0787
5 0.0806
5 0.0781
5 0.0736
5 0.0793
6 0.0918
6 0.0896
6 0.09
6 0.0856
6 0.0891
6 0.0848
6 0.0894
6 0.0917
6 0.089
6 0.0872
7 0.1003
7 0.0988
7 0.0995
7 0.0955
7 0.0947
7 0.0964
7 0.1009
7 0.099
7 0.0958
7 0.0997
8 0.1046
8 0.1068
8 0.1108
8 0.109
8 0.109
8 0.1022
8 0.1093
8 0.1046
8 0.1053
8 0.1058
9 0.1131
9 0.1081
9 0.1102
9 0.1156
9 0.1106
9 0.1126
9 0.1092
9 0.1094
9 0.1117
9 0.1183
10 0.1105
10 0.1154
10 0.1196
10 0.1161
10 0.1186
10 0.1109
10 0.1141
10 0.1141
10 0.1109
10 0.1181
1 0.0195
1 0.0198
1 0.0196
1 0.02
1 0.0197
1 0.0207
1 0.0191
1 0.0195
1 0.018
1 0.0199
2 0.0364
2 0.0365
2 0.0379
2 0.0371
2 0.0378
2 0.0404
2 0.0382
2 0.0361
2 0.0393
2 0.0357
3 0.0514
3 0.0527
3 0.0543
3 0.0528
3 0.0534
3 0.0528
3 0.0509
3 0.0521
3 0.0535
3 0.0508
4 0.0618
4 0.0701
4 0.067
4 0.0638
4 0.0637
4 0.067
4 0.065
4 0.0628
4 0.0657
4 0.0669
5 0.0796
5 0.0734
5 0.0784
5 0.0743
5 0.0761
5 0.0799
5 0.0773
5 0.0768
5 0.0789
5 0.0754
6 0.0843
6 0.0886
6 0.0859
6 0.0868
6 0.0878
6 0.0899
6 0.0897
6 0.0868
6 0.0892
6 0.0876
7 0.0946
7 0.0927
7 0.0945
7 0.0963
7 0.1004
7 0.0999
7 0.1003
7 0.0962
7 0.0988
7 0.0946
8 0.106
8 0.1033
8 0.0995
8 0.1096
8 0.1054
8 0.1084
8 0.0995
8 0.108
8 0.0999
8 0.0971
9 0.112
9 0.1066
9 0.1029
9 0.1117
9 0.1081
9 0.108
9 0.1138
9 0.1146
9 0.1107
9 0.1058
10 0.1151
10 0.11
10 0.1102
10 0.11
10 0.1104
10 0.1135
10 0.1124
10 0.1115
10 0.1086
10 0.1174
1 0.0192
1 0.0198
1 0.0185
1 0.0212
1 0.0183
1 0.0168
1 0.0199
1 0.0202
1 0.0195
1 0.0194
2 0.0366
2 0.0371
2 0.0367
2 0.0352
2 0.0383
2 0.0354
2 0.0371
2 0.0364
2 0.0355
2 0.0341
3 0.0506
3 0.0516
3 0.0486
3 0.0524
3 0.0528
3 0.0489
3 0.0498
3 0.0509
3 0.05
3 0.0535
4 0.0635
4 0.0638
4 0.0664
4 0.0624
4 0.065
4 0.06
4 0.0638
4 0.065
4 0.0629
4 0.0704
5 0.0742
5 0.0686
5 0.0762
5 0.0738
5 0.0797
5 0.0792
5 0.0769
5 0.0749
5 0.0795
5 0.0765
6 0.0908
6 0.0888
6 0.0912
6 0.092
6 0.0887
6 0.0843
6 0.0854
6 0.0864
6 0.0868
6 0.0876
7 0.0969
7 0.0974
7 0.0953
7 0.0968
7 0.0941
7 0.0938
7 0.0968
7 0.0964
7 0.0938
7 0.0977
8 0.1042
8 0.1059
8 0.1023
8 0.1005
8 0.1011
8 0.1054
8 0.1014
8 0.1104
8 0.1042
8 0.1017
9 0.107
9 0.1083
9 0.1063
9 0.109
9 0.108
9 0.1108
9 0.1093
9 0.1083
9 0.1094
9 0.1049
10 0.1076
10 0.1111
10 0.1096
10 0.1097
10 0.1077
10 0.1133
10 0.1091
10 0.1119
10 0.1066
10 0.1111
1 0.0172
1 0.0198
1 0.0197
1 0.0202
1 0.0179
1 0.0174
1 0.0182
1 0.0191
1 0.021
1 0.0205
2 0.0362
2 0.0363
2 0.0382
2 0.0351
2 0.0381
2 0.0368
2 0.0386
2 0.0364
2 0.0367
2 0.0388
3 0.0521
3 0.0494
3 0.0514
3 0.0521
3 0.0521
3 0.0504
3 0.0524
3 0.0515
3 0.0538
3 0.0503
4 0.0687
4 0.0646
4 0.0658
4 0.0688
4 0.0647
4 0.0637
4 0.0679
4 0.0647
4 0.0643
4 0.0643
5 0.0788
5 0.0778
5 0.078
5 0.0749
5 0.0776
5 0.078
5 0.0793
5 0.0826
5 0.0781
5 0.0755
6 0.0925
6 0.0903
6 0.0899
6 0.0899
6 0.0879
6 0.0896
6 0.0891
6 0.0838
6 0.0913
6 0.0909
7 0.0967
7 0.0957
7 0.0993
7 0.1005
7 0.1001
7 0.0993
7 0.0983
7 0.0926
7 0.1
7 0.0961
8 0.1022
8 0.101
8 0.1053
8 0.1115
8 0.1051
8 0.1026
8 0.1031
8 0.1029
8 0.101
8 0.1059
9 0.116
9 0.1061
9 0.1091
9 0.1135
9 0.1067
9 0.1033
9 0.1099
9 0.1079
9 0.1084
9 0.1088
10 0.1112
10 0.1136
10 0.111
10 0.107
10 0.1146
10 0.1141
10 0.1091
10 0.1095
10 0.1124
10 0.1141
1 0.0201
1 0.0204
1 0.0219
1 0.0183
1 0.0192
1 0.02
1 0.0202
1 0.0213
1 0.0187
1 0.02
2 0.0368
2 0.0364
2 0.037
2 0.0359
2 0.0377
2 0.0379
2 0.0382
2 0.0383
2 0.0375
2 0.0373
3 0.0526
3 0.0489
3 0.0554
3 0.0515
3 0.0527
3 0.0519
3 0.0513
3 0.0545
3 0.0555
3 0.0498
4 0.0682
4 0.066
4 0.0668
4 0.0638
4 0.0635
4 0.0639
4 0.0639
4 0.0674
4 0.0635
4 0.0685
5 0.0808
5 0.0783
5 0.076
5 0.0772
5 0.0784
5 0.078
5 0.0782
5 0.079
5 0.0739
5 0.0801
6 0.0914
6 0.0914
6 0.088
6 0.0908
6 0.0866
6 0.0835
6 0.0942
6 0.0829
6 0.0885
6 0.0928
7 0.0982
7 0.1002
7 0.0949
7 0.099
7 0.0954
7 0.0988
7 0.1005
7 0.0981
7 0.099
7 0.0985
8 0.1064
8 0.1022
8 0.1091
8 0.108
8 0.103
8 0.111
8 0.101
8 0.1114
8 0.1005
8 0.1046
9 0.1132
9 0.1118
9 0.1078
9 0.108
9 0.1089
9 0.1118
9 0.1117
9 0.1097
9 0.1097
9 0.1076
10 0.1128
10 0.1105
10 0.1143
10 0.1122
10 0.1171
10 0.1104
10 0.108
10 0.1046
10 0.1039
10 0.1115
};
\end{axis}

\end{tikzpicture}

%% file: Figures/degree_edges_spatial_uniform.tex
\begin{tikzpicture}

\definecolor{color0}{rgb}{0.12156862745098,0.466666666666667,0.705882352941177}

\begin{axis}[
height=6cm,
tick align=outside,
tick pos=left,
width=8cm,
x grid style={white!69.0196078431373!black},
xlabel={in-degree},
xmin=0, xmax=17.8,
xtick style={color=black},
y grid style={white!69.0196078431373!black},
ylabel={\(\displaystyle \mathcal{M}_i^{ld}(W)\)},
ymin=0, ymax=0.23033,
ytick style={color=black}
]
\addplot [draw=color0, fill=color0, mark=*, only marks, opacity=0.8, scatter]
table{%
x  y
16 0.2152
7 0.0804
12 0.1442
11 0.151
14 0.1938
15 0.1986
10 0.1214
8 0.0934
9 0.129
8 0.0926
13 0.1686
6 0.0762
10 0.1176
11 0.1312
11 0.1644
12 0.1432
8 0.0958
9 0.0954
11 0.1664
9 0.0984
10 0.1472
7 0.1128
11 0.1384
10 0.132
6 0.0752
10 0.1366
9 0.1032
6 0.0844
10 0.117
14 0.1602
14 0.1628
5 0.0742
13 0.1334
15 0.1988
13 0.1312
13 0.155
12 0.1342
7 0.0894
11 0.1374
11 0.1322
14 0.1656
9 0.1296
11 0.1658
9 0.1096
12 0.1454
9 0.1144
9 0.115
9 0.1218
7 0.0864
9 0.1072
8 0.104
13 0.1658
9 0.1062
9 0.1192
10 0.1226
13 0.1672
7 0.1144
10 0.1488
10 0.11
10 0.1308
7 0.0998
12 0.1546
10 0.1426
11 0.1508
12 0.1254
9 0.128
10 0.1344
8 0.0982
11 0.1274
11 0.1284
6 0.0778
11 0.1662
13 0.1564
11 0.1164
12 0.1584
12 0.1564
9 0.0978
12 0.1546
7 0.0922
5 0.0666
9 0.1098
12 0.1458
15 0.189
9 0.115
10 0.1188
11 0.1408
2 0.0518
6 0.0828
5 0.0706
7 0.1122
6 0.077
15 0.1718
11 0.1376
11 0.1376
7 0.099
12 0.1532
11 0.1324
11 0.1178
9 0.111
8 0.0988
10 0.1234
11 0.1306
8 0.0962
15 0.1868
14 0.1656
13 0.1568
11 0.1208
14 0.161
6 0.0952
10 0.1454
9 0.1132
7 0.085
7 0.089
8 0.1256
6 0.0854
14 0.195
13 0.1474
14 0.186
10 0.1368
13 0.16
10 0.1436
5 0.0632
12 0.1398
10 0.113
9 0.124
11 0.1308
10 0.1392
14 0.1588
10 0.1176
10 0.118
10 0.1126
11 0.125
6 0.0768
13 0.1486
12 0.1556
11 0.1442
11 0.148
10 0.1194
8 0.104
9 0.1294
9 0.1214
5 0.0606
11 0.147
13 0.1542
11 0.1316
11 0.1388
14 0.1482
8 0.1082
6 0.0804
9 0.0938
11 0.164
8 0.1186
12 0.1296
8 0.1126
12 0.1466
16 0.1758
10 0.122
11 0.1478
11 0.1258
9 0.1256
13 0.1392
6 0.0942
4 0.059
6 0.0716
9 0.1228
16 0.1664
10 0.1258
15 0.1926
3 0.0554
10 0.1302
16 0.2008
17 0.1934
13 0.1806
6 0.078
8 0.0942
8 0.1122
5 0.0834
8 0.1026
14 0.1616
7 0.0986
7 0.0962
14 0.1716
10 0.1476
12 0.158
6 0.082
7 0.1014
9 0.1172
9 0.1242
10 0.1408
17 0.1744
10 0.102
7 0.0902
9 0.137
8 0.1016
5 0.0818
10 0.1418
8 0.1086
13 0.175
9 0.1136
6 0.1018
13 0.1752
12 0.1308
13 0.141
9 0.1136
7 0.0846
7 0.0734
12 0.1618
6 0.0646
6 0.0824
8 0.1024
10 0.1182
13 0.1434
7 0.0878
4 0.0514
9 0.144
10 0.1356
7 0.102
12 0.1332
13 0.1804
9 0.1196
13 0.1724
6 0.0802
11 0.1552
8 0.1208
9 0.1254
7 0.107
13 0.1448
12 0.1368
5 0.0672
5 0.0592
7 0.0744
12 0.1666
12 0.1554
12 0.144
13 0.1784
6 0.0738
10 0.1192
11 0.119
3 0.0548
4 0.0516
10 0.1248
16 0.1772
11 0.1282
14 0.1946
9 0.1222
10 0.1418
15 0.1868
10 0.1438
11 0.1244
10 0.1342
13 0.1888
13 0.1776
11 0.1594
10 0.1156
17 0.2022
10 0.1486
11 0.1358
11 0.1654
10 0.1262
11 0.157
15 0.2008
8 0.1062
8 0.1102
1 0.0386
9 0.1226
17 0.1804
13 0.1474
6 0.0708
10 0.113
8 0.084
10 0.1234
5 0.061
12 0.1518
7 0.0866
10 0.1222
13 0.1748
5 0.07
9 0.1258
13 0.1782
12 0.1616
9 0.1232
10 0.1224
7 0.0854
8 0.0984
9 0.1064
14 0.192
12 0.1716
11 0.1476
14 0.1938
9 0.1194
11 0.1502
10 0.109
9 0.1126
13 0.143
11 0.1276
15 0.182
9 0.114
9 0.1286
12 0.1254
8 0.114
11 0.1368
16 0.2208
7 0.0876
15 0.1698
13 0.1602
12 0.1462
4 0.057
11 0.146
10 0.112
8 0.0982
10 0.1318
11 0.1342
12 0.161
9 0.1224
6 0.0852
12 0.1244
11 0.1622
11 0.1256
10 0.117
5 0.0764
5 0.0674
6 0.087
5 0.074
11 0.133
12 0.167
9 0.1148
11 0.1222
16 0.1802
10 0.1164
12 0.1244
9 0.1136
9 0.126
10 0.1294
6 0.0754
8 0.1246
6 0.0756
7 0.098
9 0.1408
8 0.1112
9 0.1088
15 0.2212
15 0.1744
11 0.1326
5 0.0974
16 0.1908
14 0.2062
12 0.137
16 0.1846
11 0.129
8 0.1006
10 0.1174
10 0.103
12 0.1822
9 0.1184
7 0.091
8 0.1002
7 0.0962
12 0.137
8 0.1062
8 0.0992
4 0.0648
9 0.116
11 0.145
13 0.1888
9 0.098
11 0.1282
10 0.1202
13 0.1754
5 0.0824
10 0.1278
6 0.0792
13 0.1862
15 0.2072
11 0.1422
7 0.0926
7 0.1006
13 0.1394
17 0.2168
14 0.1818
11 0.131
10 0.1452
6 0.0726
7 0.0842
11 0.12
10 0.119
6 0.0966
12 0.147
11 0.1564
6 0.0796
10 0.1476
14 0.1856
10 0.1246
11 0.1344
13 0.185
7 0.0946
11 0.1332
7 0.0974
9 0.109
9 0.1106
15 0.2074
6 0.0744
6 0.0694
17 0.1984
7 0.1018
12 0.1438
8 0.0926
10 0.1474
9 0.119
13 0.1728
8 0.0992
6 0.0816
11 0.1758
7 0.0894
12 0.1408
9 0.1216
15 0.1684
15 0.1884
8 0.0686
6 0.082
7 0.0806
9 0.0936
9 0.086
10 0.1604
6 0.0906
14 0.21
11 0.1436
14 0.1504
15 0.2086
12 0.1592
9 0.0848
13 0.1458
13 0.1508
11 0.1092
8 0.0896
17 0.1922
13 0.181
12 0.1234
13 0.1384
15 0.2142
10 0.1128
8 0.0862
14 0.1406
9 0.1078
9 0.0906
11 0.1418
3 0.0426
13 0.1246
13 0.126
8 0.0806
11 0.1144
6 0.0816
7 0.1078
14 0.1896
6 0.0618
5 0.0632
4 0.0568
11 0.1058
12 0.152
11 0.114
12 0.1366
7 0.0864
13 0.1788
5 0.051
5 0.0558
6 0.0766
12 0.1404
6 0.0716
13 0.1764
14 0.2058
10 0.1158
8 0.0974
11 0.1702
10 0.145
13 0.1558
5 0.0718
12 0.1536
13 0.1516
14 0.1662
12 0.1524
5 0.0544
11 0.153
7 0.0954
6 0.0662
13 0.1144
11 0.157
5 0.0542
8 0.0862
7 0.0878
8 0.0884
14 0.1708
10 0.134
8 0.1014
10 0.1428
13 0.143
10 0.1364
12 0.17
10 0.103
6 0.0844
11 0.1382
13 0.1558
};
\end{axis}

\end{tikzpicture}

%% file: Figures/degree_edges_spatial_gaussian.tex
\begin{tikzpicture}

\definecolor{color0}{rgb}{0.12156862745098,0.466666666666667,0.705882352941177}

\begin{axis}[
height=6cm,
tick align=outside,
tick pos=left,
width=8cm,
x grid style={white!69.0196078431373!black},
xlabel={in-degree},
xmin=0, xmax=19.95,
xtick style={color=black},
y grid style={white!69.0196078431373!black},
ylabel={\(\displaystyle \mathcal{M}_i^{ld}(W)\)},
ymin=0, ymax=0.26949,
ytick style={color=black}
]
\addplot [draw=color0, fill=color0, mark=*, only marks, opacity=0.8, scatter]
table{%
x  y
7 0.0688
11 0.1544
3 0.0454
5 0.0602
5 0.0604
1 0.0324
14 0.1828
10 0.1092
11 0.1118
1 0.0334
9 0.1232
10 0.1102
10 0.111
17 0.258
11 0.1862
10 0.1108
8 0.108
15 0.2372
14 0.158
6 0.0694
10 0.1488
12 0.199
1 0.0346
11 0.123
13 0.1498
7 0.072
13 0.2224
9 0.096
10 0.104
14 0.1822
6 0.0638
15 0.149
5 0.058
5 0.0562
12 0.1236
10 0.0902
10 0.1088
6 0.0552
8 0.0858
11 0.1484
13 0.1338
11 0.1386
14 0.1542
3 0.0412
6 0.0652
11 0.1216
9 0.1026
12 0.164
11 0.126
9 0.1206
14 0.1512
9 0.0798
8 0.0952
16 0.1648
10 0.1046
9 0.1288
10 0.089
11 0.1314
11 0.1144
14 0.1976
14 0.14
6 0.0682
9 0.0928
14 0.1704
12 0.1716
8 0.1188
9 0.1122
9 0.0898
14 0.142
15 0.184
9 0.1378
14 0.132
15 0.1634
3 0.0458
9 0.0954
14 0.1358
9 0.088
10 0.129
16 0.227
10 0.0942
12 0.1438
11 0.1044
11 0.1368
13 0.2056
5 0.0468
11 0.1464
6 0.0798
10 0.1046
11 0.1224
14 0.1408
6 0.0678
2 0.0398
13 0.151
11 0.1038
15 0.1584
13 0.2034
9 0.1196
10 0.1028
7 0.0698
14 0.1384
13 0.1784
7 0.0838
9 0.0982
8 0.101
9 0.1008
3 0.0474
10 0.1684
11 0.1554
8 0.0998
12 0.1094
3 0.0484
13 0.191
13 0.1876
10 0.0828
11 0.1252
14 0.161
16 0.1664
9 0.0948
4 0.0538
13 0.1744
10 0.1588
15 0.1586
3 0.0478
10 0.11
8 0.078
6 0.0802
10 0.1158
12 0.1356
8 0.1214
8 0.08
10 0.135
8 0.0872
9 0.0916
10 0.102
15 0.1988
11 0.1762
3 0.05
10 0.0914
8 0.0918
14 0.1584
9 0.1104
13 0.1912
10 0.161
6 0.0764
14 0.152
9 0.1134
15 0.157
12 0.168
13 0.111
11 0.1338
12 0.172
13 0.1974
7 0.0914
12 0.1532
14 0.1464
8 0.1232
8 0.1066
11 0.1066
4 0.0492
9 0.1068
17 0.189
9 0.124
13 0.189
12 0.127
12 0.1224
11 0.1516
13 0.1704
6 0.0728
12 0.1272
9 0.153
5 0.0608
11 0.1042
14 0.175
12 0.144
12 0.1436
11 0.122
9 0.1052
7 0.0754
16 0.2128
5 0.0572
14 0.1772
12 0.156
8 0.0768
4 0.0626
12 0.194
3 0.0474
12 0.138
2 0.0454
12 0.1468
9 0.092
16 0.2282
12 0.1386
8 0.078
10 0.1284
9 0.0842
8 0.0758
7 0.0824
11 0.1102
6 0.0618
15 0.1796
14 0.1782
10 0.1088
6 0.0652
18 0.2194
13 0.1524
15 0.1762
13 0.1708
2 0.038
12 0.1572
8 0.0882
11 0.166
13 0.2022
9 0.092
7 0.07
12 0.1318
11 0.1404
14 0.169
5 0.066
9 0.1072
7 0.1012
3 0.047
14 0.1598
10 0.1662
16 0.1958
9 0.1114
16 0.1914
11 0.152
6 0.097
9 0.084
1 0.0366
10 0.1252
12 0.1894
3 0.0472
9 0.0942
10 0.0916
9 0.1242
4 0.0528
4 0.0512
5 0.061
8 0.0898
8 0.087
16 0.1984
14 0.213
3 0.0468
16 0.205
5 0.055
2 0.0438
14 0.1946
10 0.1168
8 0.0894
12 0.131
14 0.1892
7 0.073
11 0.134
12 0.1592
16 0.2404
14 0.178
5 0.0556
12 0.175
11 0.1196
9 0.117
13 0.149
16 0.1878
3 0.047
18 0.22
11 0.106
17 0.2402
19 0.2422
15 0.2132
19 0.2516
3 0.0426
8 0.0912
8 0.0748
12 0.1796
8 0.118
14 0.1664
12 0.1744
6 0.0852
17 0.2316
13 0.172
10 0.1454
8 0.0706
13 0.1444
7 0.0962
15 0.222
3 0.0436
7 0.0736
8 0.1122
11 0.1284
8 0.0764
8 0.082
7 0.0918
5 0.0516
3 0.044
11 0.0938
10 0.1108
10 0.0888
9 0.1156
6 0.0608
12 0.1506
12 0.139
9 0.09
14 0.1424
8 0.0914
11 0.1158
8 0.0932
7 0.1078
11 0.1318
9 0.1424
12 0.1406
0 0.0302
10 0.1218
9 0.1206
6 0.067
19 0.2396
12 0.151
7 0.0642
13 0.1918
13 0.1312
13 0.148
5 0.0592
8 0.083
13 0.1546
8 0.1022
12 0.1684
6 0.068
3 0.0488
15 0.151
13 0.1608
12 0.1554
13 0.1354
13 0.2082
11 0.1384
10 0.1004
16 0.1552
17 0.1948
9 0.1304
8 0.081
13 0.147
9 0.1176
8 0.079
11 0.149
6 0.0664
7 0.115
9 0.0966
14 0.1686
8 0.1156
15 0.224
6 0.0664
4 0.0544
13 0.1242
9 0.0996
7 0.1034
14 0.132
15 0.2298
8 0.0874
12 0.2042
15 0.1996
6 0.077
9 0.1236
12 0.1628
7 0.0994
8 0.1166
9 0.089
13 0.1628
4 0.0594
13 0.1994
13 0.1518
10 0.1042
16 0.1722
2 0.043
5 0.0566
13 0.1954
7 0.082
16 0.154
7 0.0794
13 0.1478
12 0.2036
12 0.1296
10 0.1044
5 0.0576
14 0.1626
12 0.1668
5 0.0616
16 0.1972
9 0.1278
8 0.087
6 0.082
6 0.0736
2 0.0404
18 0.2012
13 0.1982
10 0.0938
10 0.1196
11 0.1128
15 0.169
3 0.0438
5 0.0604
3 0.0486
14 0.149
13 0.1416
6 0.0682
10 0.1032
14 0.1528
8 0.068
8 0.0852
11 0.1814
10 0.1514
11 0.1092
13 0.1232
9 0.0862
12 0.1108
11 0.1176
10 0.1234
10 0.1208
9 0.0858
8 0.0888
9 0.106
5 0.0498
1 0.0348
11 0.115
10 0.1176
9 0.1334
11 0.1202
12 0.1118
12 0.1716
12 0.146
14 0.1406
12 0.152
2 0.0354
13 0.1458
14 0.1506
11 0.1086
8 0.0772
9 0.1238
12 0.1504
7 0.0872
15 0.1718
8 0.1112
13 0.1388
12 0.17
3 0.0376
3 0.039
13 0.1448
4 0.0558
14 0.1468
7 0.062
11 0.1052
6 0.0614
13 0.1912
13 0.193
16 0.2002
10 0.1092
9 0.1304
4 0.0414
10 0.1422
11 0.1134
4 0.0502
13 0.1322
13 0.1878
7 0.0644
0 0.0282
13 0.1434
16 0.1916
9 0.1226
14 0.1426
10 0.1182
15 0.216
17 0.2426
12 0.1948
13 0.1576
13 0.1556
14 0.1388
5 0.051
10 0.1682
16 0.2056
12 0.1568
13 0.2044
14 0.142
12 0.1544
3 0.0388
13 0.1736
9 0.1522
13 0.1546
11 0.1196
11 0.1404
8 0.0704
16 0.2394
7 0.0758
8 0.0682
3 0.0394
4 0.0408
10 0.1408
10 0.0976
13 0.1476
12 0.1386
5 0.0536
4 0.0492
10 0.1502
3 0.039
};
\end{axis}

\end{tikzpicture}

%% file: Figures/degree_edges_small_world.tex
\begin{tikzpicture}

\definecolor{color0}{rgb}{0.12156862745098,0.466666666666667,0.705882352941177}

\begin{axis}[
height=6cm,
tick align=outside,
tick pos=left,
width=8cm,
x grid style={white!69.0196078431373!black},
xlabel={in-degree},
xmin=0, xmax=14.4,
xtick style={color=black},
y grid style={white!69.0196078431373!black},
ylabel={\(\displaystyle \mathcal{M}_i^{ld}(W)\)},
ymin=0, ymax=0.17015,
ytick style={color=black}
]
\addplot [draw=color0, fill=color0, mark=*, only marks, opacity=0.8, scatter]
table{%
x  y
12 0.1408
9 0.1138
10 0.1184
10 0.1308
9 0.1196
9 0.1188
9 0.1162
10 0.1262
14 0.1662
10 0.124
11 0.1326
9 0.1168
10 0.124
12 0.1386
8 0.1064
10 0.1214
9 0.1142
10 0.1204
8 0.1082
9 0.1074
11 0.1396
10 0.1234
10 0.1228
8 0.1024
9 0.1144
11 0.1316
10 0.1212
8 0.1072
12 0.1438
9 0.114
8 0.1034
11 0.1408
11 0.1338
11 0.1312
11 0.1296
11 0.1286
9 0.115
10 0.1212
7 0.0922
10 0.1168
9 0.1126
10 0.1174
10 0.1208
13 0.143
9 0.1098
11 0.1346
8 0.106
10 0.1224
10 0.1204
9 0.1162
8 0.1116
11 0.1274
13 0.1546
10 0.121
11 0.1272
9 0.1136
8 0.1052
11 0.13
8 0.1024
10 0.1206
10 0.121
10 0.1272
9 0.1104
12 0.1412
10 0.1258
8 0.1112
13 0.1574
9 0.1188
8 0.0994
10 0.1212
12 0.1446
10 0.1168
9 0.111
9 0.1204
12 0.1324
10 0.1196
10 0.1166
11 0.1266
8 0.1074
11 0.1286
11 0.133
10 0.1156
9 0.111
12 0.1438
12 0.1394
9 0.1142
11 0.1258
12 0.1382
10 0.1226
11 0.135
11 0.1262
10 0.123
10 0.1176
10 0.126
9 0.1108
8 0.1064
11 0.1348
10 0.119
10 0.1208
10 0.1186
10 0.1258
9 0.118
11 0.141
11 0.1294
11 0.1374
10 0.1318
12 0.15
10 0.1176
9 0.1254
9 0.1196
8 0.1144
11 0.1318
10 0.1254
10 0.1314
9 0.1194
12 0.1478
8 0.113
9 0.1184
11 0.1422
9 0.117
8 0.1164
11 0.139
10 0.1298
10 0.128
10 0.1254
8 0.1106
9 0.1144
10 0.1294
10 0.1282
12 0.1402
9 0.1178
9 0.1218
11 0.1362
13 0.1544
10 0.1282
11 0.1296
11 0.1276
9 0.119
12 0.1504
11 0.1342
10 0.1242
10 0.1328
10 0.1316
9 0.1222
11 0.1348
9 0.1232
11 0.1308
11 0.1428
12 0.142
10 0.1226
9 0.1194
9 0.1198
9 0.1194
9 0.1196
9 0.1282
9 0.1234
12 0.1494
10 0.1382
10 0.1292
11 0.1452
10 0.1326
9 0.1246
9 0.122
10 0.1332
10 0.1268
9 0.1152
7 0.0964
9 0.1168
10 0.1294
11 0.1364
11 0.1358
10 0.1288
11 0.138
10 0.1252
12 0.1466
9 0.1198
11 0.1316
11 0.1342
11 0.1378
10 0.1278
12 0.1456
12 0.1438
11 0.1348
11 0.137
10 0.1282
8 0.113
8 0.1098
9 0.1228
7 0.1022
10 0.1298
10 0.1308
11 0.1384
10 0.1324
12 0.1536
11 0.1366
8 0.108
10 0.1302
11 0.1474
8 0.1058
8 0.1072
11 0.1436
9 0.1214
11 0.1466
10 0.1258
8 0.1036
11 0.1346
12 0.141
10 0.1264
12 0.1476
10 0.1232
9 0.1184
10 0.1272
9 0.114
9 0.1126
9 0.1098
10 0.1226
11 0.1364
10 0.1278
11 0.1338
8 0.104
10 0.127
10 0.1258
10 0.121
8 0.1028
10 0.1224
10 0.1208
11 0.1264
11 0.1372
10 0.1256
12 0.1406
10 0.124
9 0.1076
11 0.1336
10 0.1226
9 0.1166
10 0.1232
11 0.142
10 0.1248
9 0.12
10 0.1266
11 0.1358
9 0.1192
9 0.1136
8 0.1058
10 0.1336
9 0.119
12 0.1496
10 0.1266
10 0.1242
10 0.1264
8 0.108
9 0.1164
10 0.1324
12 0.149
6 0.0872
10 0.1172
11 0.1344
9 0.114
9 0.116
12 0.144
12 0.15
10 0.1318
11 0.1396
10 0.1246
11 0.1348
10 0.1294
13 0.149
10 0.1234
11 0.1356
10 0.1194
9 0.115
9 0.1134
8 0.112
9 0.1192
11 0.1308
10 0.1266
11 0.1304
10 0.1262
11 0.1342
9 0.1094
13 0.1586
11 0.1318
8 0.1006
12 0.1486
9 0.1128
9 0.1144
7 0.1036
11 0.1304
11 0.1398
11 0.1344
10 0.1298
11 0.136
10 0.1344
8 0.1086
9 0.1202
11 0.137
10 0.1258
10 0.1272
10 0.1334
9 0.118
9 0.116
10 0.1294
9 0.1154
10 0.1266
9 0.119
9 0.1148
10 0.1206
12 0.15
11 0.1366
10 0.1258
12 0.1468
9 0.1184
9 0.1124
8 0.1038
9 0.1148
10 0.121
10 0.1258
13 0.152
9 0.1182
11 0.1306
10 0.1294
9 0.1208
9 0.1098
9 0.117
10 0.1192
9 0.1158
9 0.115
9 0.1176
10 0.131
10 0.128
10 0.1312
10 0.1306
13 0.1488
11 0.1386
11 0.132
11 0.14
11 0.1292
10 0.124
8 0.1142
8 0.1036
9 0.1246
9 0.1076
10 0.1162
10 0.1272
12 0.1376
11 0.1356
9 0.1174
9 0.1196
11 0.1358
10 0.1272
10 0.1268
10 0.123
10 0.1244
14 0.1654
8 0.1032
8 0.1102
10 0.1272
11 0.1332
11 0.1294
10 0.1252
9 0.1112
9 0.1162
11 0.136
10 0.1202
10 0.1176
11 0.1262
10 0.1286
11 0.1324
10 0.1262
10 0.1274
10 0.1288
9 0.1176
9 0.1202
10 0.1288
11 0.1348
10 0.1272
8 0.0996
10 0.1254
10 0.1274
10 0.1256
10 0.122
10 0.122
9 0.1152
11 0.143
10 0.129
10 0.1266
12 0.148
10 0.1276
10 0.124
9 0.1194
11 0.1304
10 0.1248
10 0.1252
10 0.1274
9 0.1184
9 0.1196
12 0.1432
11 0.1448
11 0.1412
10 0.126
9 0.1064
10 0.119
11 0.123
10 0.1162
10 0.1204
12 0.1348
10 0.1146
9 0.1082
10 0.1168
10 0.1262
9 0.1106
10 0.1198
12 0.1408
13 0.15
10 0.1256
11 0.1396
10 0.1204
9 0.116
10 0.1242
10 0.126
9 0.111
9 0.1168
9 0.115
11 0.133
10 0.1128
9 0.1088
11 0.1328
11 0.1336
9 0.109
7 0.097
7 0.103
8 0.1054
10 0.124
9 0.1136
9 0.114
7 0.0954
13 0.1432
12 0.1328
11 0.1224
12 0.1358
10 0.124
10 0.1196
11 0.123
7 0.095
9 0.1126
9 0.1192
9 0.1154
10 0.1226
9 0.115
9 0.1186
11 0.1364
11 0.1326
8 0.1026
8 0.1096
7 0.1006
10 0.1226
11 0.13
10 0.1284
10 0.1298
12 0.1438
10 0.123
10 0.127
11 0.1386
10 0.1204
10 0.1284
13 0.1488
12 0.1404
11 0.1246
12 0.1362
11 0.1304
12 0.139
10 0.1246
10 0.1152
9 0.1106
9 0.1036
10 0.1202
9 0.108
12 0.1362
11 0.1298
10 0.1204
11 0.127
11 0.1312
8 0.103
10 0.1194
9 0.112
10 0.114
9 0.115
10 0.1254
9 0.1068
12 0.1312
10 0.1198
10 0.1212
9 0.1104
10 0.1268
9 0.1158
11 0.1308
10 0.1272
10 0.1156
10 0.12
11 0.1342
};
\end{axis}

\end{tikzpicture}

%% file: Figures/degree_edges_pref.tex
\begin{tikzpicture}

\definecolor{color0}{rgb}{0.12156862745098,0.466666666666667,0.705882352941177}

\begin{axis}[
height=6cm,
tick align=outside,
tick pos=left,
width=8cm,
x grid style={white!69.0196078431373!black},
xlabel={in-degree},
xmin=0, xmax=43.9,
xtick style={color=black},
y grid style={white!69.0196078431373!black},
ylabel={\(\displaystyle \mathcal{M}_i^{ld}(W)\)},
ymin=0, ymax=0.4576,
ytick style={color=black}
]
\addplot [draw=color0, fill=color0, mark=*, only marks, opacity=0.8, scatter]
table{%
x  y
24 0.2652
7 0.0976
12 0.1512
22 0.2384
30 0.3142
29 0.3112
15 0.1772
29 0.3212
24 0.2508
9 0.1172
7 0.1024
15 0.1772
10 0.1304
11 0.1382
9 0.1124
15 0.1668
10 0.1258
8 0.1126
10 0.13
11 0.142
10 0.1326
13 0.149
10 0.1314
9 0.1144
7 0.0974
12 0.1426
8 0.106
14 0.1756
7 0.0912
12 0.152
9 0.1102
7 0.1052
7 0.0954
7 0.1018
6 0.0934
4 0.0648
9 0.1136
7 0.095
6 0.0844
6 0.091
4 0.071
5 0.0808
6 0.0924
8 0.1052
6 0.09
4 0.0704
7 0.1004
5 0.0756
7 0.0948
4 0.0682
6 0.0858
5 0.079
7 0.104
6 0.087
7 0.0954
4 0.0716
8 0.1096
4 0.0696
5 0.0746
5 0.0862
4 0.0662
5 0.0766
6 0.0904
4 0.0702
5 0.0782
6 0.0922
4 0.069
6 0.086
5 0.0824
4 0.0688
5 0.0846
4 0.0692
5 0.0824
4 0.0654
4 0.0688
4 0.0688
6 0.084
4 0.0644
5 0.0834
5 0.0764
4 0.069
5 0.0812
4 0.07
5 0.0804
4 0.072
4 0.0672
4 0.0646
4 0.0744
5 0.0802
4 0.0754
4 0.0662
4 0.0688
4 0.074
4 0.067
4 0.0718
4 0.0672
4 0.0742
4 0.0716
4 0.067
4 0.0678
17 0.1932
12 0.1448
13 0.1572
19 0.2174
27 0.292
19 0.216
25 0.2722
19 0.216
21 0.2304
15 0.1786
22 0.237
5 0.0824
17 0.1914
12 0.1428
10 0.1232
13 0.1476
15 0.1738
21 0.2332
8 0.1076
15 0.1684
13 0.1496
8 0.107
10 0.1238
7 0.0948
7 0.0942
11 0.1318
9 0.1126
7 0.0884
6 0.0832
8 0.1012
6 0.0898
6 0.0886
6 0.0826
7 0.0936
6 0.0832
8 0.1018
5 0.0816
4 0.0668
7 0.0966
4 0.0688
10 0.1324
5 0.0782
6 0.085
8 0.1036
7 0.097
4 0.072
6 0.083
5 0.073
5 0.0754
6 0.0846
6 0.084
6 0.0892
4 0.0712
5 0.0722
4 0.0676
8 0.1052
6 0.089
5 0.0794
6 0.0884
4 0.0692
7 0.1002
5 0.0724
4 0.066
6 0.0836
5 0.0742
4 0.062
5 0.075
7 0.0982
5 0.0764
4 0.07
6 0.0806
4 0.0626
4 0.0634
4 0.0694
4 0.0634
4 0.0666
4 0.0664
4 0.068
7 0.0976
4 0.067
6 0.0824
4 0.0644
4 0.066
5 0.076
4 0.0646
4 0.0698
4 0.0634
4 0.0668
4 0.0662
4 0.0682
4 0.0636
4 0.0716
6 0.0864
4 0.0636
4 0.07
4 0.0678
4 0.0706
4 0.0632
4 0.068
4 0.0682
12 0.1404
19 0.2072
14 0.1618
14 0.158
40 0.4124
28 0.295
17 0.1868
22 0.2418
7 0.0964
12 0.1418
12 0.1586
17 0.1966
16 0.1762
16 0.1886
9 0.1122
15 0.1722
11 0.1406
7 0.099
20 0.221
6 0.0886
6 0.0828
8 0.1076
5 0.0792
10 0.13
6 0.0884
7 0.101
5 0.0748
10 0.123
8 0.1036
21 0.2284
5 0.077
10 0.1188
8 0.105
5 0.0758
6 0.0838
7 0.1006
7 0.0998
13 0.1504
6 0.0922
5 0.072
5 0.0742
11 0.1402
7 0.093
7 0.0966
5 0.0746
7 0.0964
7 0.0994
6 0.0924
9 0.118
5 0.0748
5 0.078
5 0.077
5 0.0744
4 0.066
4 0.0638
4 0.0656
6 0.089
5 0.0804
5 0.0786
6 0.0906
6 0.0882
6 0.0846
5 0.0792
6 0.0852
4 0.0682
4 0.0702
5 0.0794
4 0.0644
5 0.0812
4 0.0712
5 0.0798
5 0.0744
4 0.0664
5 0.0792
5 0.073
5 0.0808
6 0.0852
4 0.0686
4 0.0656
5 0.0826
4 0.0674
4 0.0684
4 0.0654
4 0.067
4 0.0662
4 0.0688
4 0.0644
5 0.0764
4 0.066
4 0.067
4 0.0642
4 0.0632
4 0.0658
4 0.0658
4 0.0658
4 0.0646
4 0.065
4 0.0638
4 0.0644
4 0.0686
21 0.228
7 0.0874
10 0.1192
29 0.3078
30 0.31
20 0.2156
20 0.2166
24 0.2606
15 0.1806
26 0.2918
17 0.197
12 0.149
10 0.126
19 0.2132
8 0.1148
12 0.1434
12 0.1388
13 0.157
8 0.103
15 0.171
7 0.0992
5 0.079
8 0.1076
13 0.1536
6 0.0876
11 0.1404
11 0.1364
7 0.0894
5 0.0774
8 0.1064
5 0.0784
9 0.1164
6 0.0882
10 0.1228
6 0.0816
6 0.0878
4 0.0618
6 0.0834
5 0.0736
6 0.0828
4 0.0724
8 0.0978
5 0.078
6 0.084
5 0.0766
4 0.0654
6 0.0856
5 0.0826
7 0.0984
5 0.0772
7 0.096
6 0.0874
10 0.1242
6 0.0854
5 0.0752
7 0.1006
5 0.0756
4 0.0722
7 0.0952
6 0.0904
7 0.1044
4 0.0648
4 0.0636
5 0.07
5 0.0788
4 0.0628
4 0.071
4 0.0666
4 0.0664
5 0.0792
4 0.0642
4 0.0672
4 0.067
5 0.0796
4 0.0638
4 0.0664
5 0.075
5 0.0768
4 0.0684
5 0.0778
5 0.0714
4 0.0706
5 0.0786
6 0.0856
4 0.07
4 0.068
4 0.068
5 0.0784
5 0.0738
4 0.0634
4 0.0684
5 0.0754
5 0.0806
4 0.0626
4 0.0702
4 0.066
4 0.069
4 0.0692
4 0.0646
4 0.0586
42 0.4386
19 0.2098
4 0.0756
15 0.1666
12 0.154
24 0.2676
22 0.2324
12 0.1474
7 0.0982
26 0.2864
23 0.2504
11 0.1376
18 0.1984
9 0.1168
15 0.1748
8 0.111
12 0.146
10 0.1208
14 0.1712
11 0.134
12 0.1478
9 0.115
7 0.0984
12 0.1454
12 0.1578
9 0.1164
7 0.099
7 0.1024
4 0.069
8 0.109
8 0.107
7 0.107
6 0.093
7 0.1022
10 0.1238
6 0.089
6 0.0888
6 0.0918
4 0.0684
6 0.0926
6 0.0892
10 0.1306
5 0.0748
6 0.0968
8 0.1006
4 0.066
4 0.0702
6 0.0938
8 0.1128
8 0.1106
5 0.0758
6 0.0848
6 0.0894
7 0.095
5 0.0778
5 0.074
5 0.081
5 0.0844
4 0.0648
6 0.0936
5 0.0782
4 0.0706
6 0.0868
4 0.068
5 0.0764
4 0.0668
6 0.09
5 0.0732
5 0.0776
5 0.0786
4 0.0678
4 0.0646
5 0.083
4 0.071
6 0.0876
4 0.0704
5 0.0786
5 0.075
4 0.067
4 0.0724
4 0.0692
4 0.0704
4 0.071
5 0.0862
4 0.0722
6 0.0898
6 0.0936
5 0.078
4 0.0666
4 0.069
4 0.0768
4 0.0742
4 0.0702
4 0.0738
5 0.077
4 0.0708
4 0.0678
5 0.0738
4 0.0666
4 0.0676
};
\end{axis}

\end{tikzpicture}

%% file: Figures/criticality_pv.tex
\begin{tikzpicture}

\begin{axis}[
legend cell align={left},
legend style={fill opacity=0.8, draw opacity=1, text opacity=1, draw=white!80!black},
tick align=outside,
tick pos=left,
x grid style={white!69.0196078431373!black},
xlabel={\(\displaystyle |V_v|\)},
xmin=1, xmax=100,
xtick style={color=black},
y grid style={white!69.0196078431373!black},
ylabel={Probability of being critical},
ymin=0, ymax=1,
ytick style={color=black}
]
\addplot [semithick, black, mark=*, mark size=3, mark options={solid}]
table {%
1 1
2 0.499695
3 0.49912
4 0.3757175
5 0.370336
10 0.243031
15 0.195378666666667
20 0.1661405
25 0.1468236
30 0.133200333333333
35 0.121684571428571
40 0.11413225
45 0.107322
50 0.1008852
55 0.0958892727272727
60 0.092337
65 0.0893755384615385
70 0.0865047142857143
75 0.0840948
80 0.08335475
85 0.0811074117647059
90 0.0804631111111111
95 0.0797926315789474
100 0.0806295
};
\addlegendentry{$\mathcal{M}^{pv_r}_v(W)$ for $v \in V_v$}
\addplot [semithick, black, dotted, mark=triangle*, mark size=3, mark options={solid,rotate=180}]
table {%
1 0
2 0.0399327551020408
3 0.000150412371134021
4 0.0308444791666667
5 0.00604547368421053
10 0.0238541111111111
15 0.0280365882352941
20 0.033275
25 0.0362729333333333
30 0.0401587142857143
35 0.0421349230769231
40 0.0452485
45 0.0483883636363636
50 0.0501004
55 0.0527482222222222
60 0.0554305
65 0.0578962857142857
70 0.0607926666666667
75 0.0633716
80 0.067133
85 0.0693513333333333
90 0.072824
95 0.0762
};
\addlegendentry{$\mathcal{M}^{pv_r}_v(W)$ for $v \in V_d$}
\addplot [semithick, black, dashed, mark=square*, mark size=3, mark options={solid}]
table {%
1 1
2 0.500715
3 0.479789999999753
4 0.3728725
5 0.33986599999993
10 0.22333599999999
15 0.176212000000041
20 0.148025499999987
25 0.131736800000026
30 0.117635333333343
35 0.110518285714265
40 0.103538999999991
45 0.0982155555555513
50 0.0932934000000211
55 0.0903903636363596
60 0.0872069999999996
65 0.0848209230768995
70 0.0845825714285611
75 0.0826755999999958
80 0.0813732499999951
85 0.0800230588235297
90 0.0807906666666669
95 0.0790699999999885
100 0.0799656000000126
};
\addlegendentry{$\mathcal{M}^{pv}_v(W)$ for $v \in V_v$}
\addplot [semithick, black, dash pattern=on 1pt off 3pt on 3pt off 3pt, mark=asterisk, mark size=3, mark options={solid}]
table {%
1 0
2 0.0400546938775542
3 0.0135848453608246
4 0.0310359374999991
5 0.0287141052631578
10 0.0398856666666657
15 0.044757764705883
20 0.0493918749999992
25 0.0520442666666683
30 0.0543058571428559
35 0.0576012307692304
40 0.0585800000000009
45 0.0604132727272743
50 0.0619108000000002
55 0.0641586666666675
60 0.0652304999999996
65 0.0670194285714288
70 0.0696580000000002
75 0.0710096000000013
80 0.0727015000000001
85 0.0737120000000006
90 0.0767320000000004
95 0.077734000000001
};
\addlegendentry{$\mathcal{M}^{pv}_v(W)$ for $v \in V_d$}
\end{axis}

\end{tikzpicture}